\documentclass[journal,10pt,doublespace,doublecolumn]{IEEEtran}
\usepackage{amsmath,amsfonts,amsthm,amssymb}
\usepackage{algorithm}
\usepackage{array}
\usepackage{textcomp}
\usepackage{balance}
\usepackage{stfloats}
\usepackage{color}
\usepackage{url}
\usepackage{verbatim}
\usepackage{graphicx}
\usepackage{cite}
\usepackage{algpseudocode}
\usepackage{subfigure}
\usepackage[dvipsnames]{xcolor}
\usepackage{multirow}
\usepackage{xspace}
\usepackage{bbm}
\usepackage{mathrsfs}
\newcommand{\Ac}{\mathcal{A}}

\newcommand{\Fc}{\mathcal{F}}
\newcommand{\Gc}{\mathcal{G}}
\newcommand{\Hc}{\mathcal{H}}
\newcommand{\Ic}{\mathcal{I}}

\newcommand{\Oc}{\mathcal{O}}
\newcommand{\Pc}{\mathcal{P}}

\newcommand{\Rc}{\mathcal{R}}

\newcommand{\Tc}{\mathcal{T}}

\newcommand{\Vc}{\mathcal{V}}

\newcommand{\Pv}{{\bf P}}

\newcommand{\kh}{{\hat{k}}}

\def\a{\alpha}
\def\b{\beta}

\def\e{\epsilon}
\def\eps{\epsilon}
\def\l{\lambda}

\newcommand{\Norm}{\mathcal{N}}

\def\textiid{i.i.d.\@\xspace}
\newcommand\iid{\ifmmode\text{ i.i.d. } \else \textiid \fi}

\newcommand{\ind}{\boldsymbol{1}}

\allowdisplaybreaks
\newtheorem{theorem}{Theorem}
\newtheorem{lemma}{Lemma}

\newtheorem{assumption}{Assumption}

\newtheorem{corollary}{Corollary}
\newtheorem{remark}{Remark}
\newtheorem{propo}{Proposition}

\allowdisplaybreaks

\begin{document}

\title{On Large-Scale Multiple Testing Over Networks: An Asymptotic Approach}

\author{\IEEEauthorblockN{Mehrdad Pournaderi,~\IEEEmembership{Student Member,~IEEE,} and Yu Xiang,~\IEEEmembership{Member,~IEEE}}\\
%\IEEEauthorblockA{University of Utah\\
                    %  50 Central Campus Dr 2110, Salt Lake City, USA\\
                    %  Email: \{mehrdad.pournaderi, yu.xiang\}@utah.edu}
\thanks{M.~Pournaderi and Y.~Xiang are with the Electrical and Computer Engineering Department, University of Utah, Salt Lake City, 
UT, 84112, USA (e-mail: \{m.pournaderi, yu.xiang\}@utah.edu)}% <-this % stops a space
}

\maketitle
\begin{abstract}
This work concerns developing communication- and computation-efficient methods for large-scale multiple testing over networks, which is of interest to many practical applications. We take an asymptotic approach and propose two methods, proportion-matching and greedy aggregation, tailored to distributed settings. The proportion-matching method achieves the global BH performance yet only requires a one-shot communication of the (estimated) proportion of true null hypotheses as well as the number of p-values at each node. By focusing on the asymptotic optimal power, we go beyond the BH procedure by providing an explicit characterization of the asymptotic optimal solution. This leads to the greedy aggregation method that effectively approximates the optimal rejection regions at each node, while computation efficiency comes from the greedy-type approach naturally. {Moreover, for both methods, we provide the rate of convergence for both the FDR and power.} Extensive numerical results over a variety of challenging settings are provided to support our theoretical findings.

\end{abstract}

\begin{IEEEkeywords}
Distributed multiple testing, asymptotic FDR control, communication-efficient methods, heterogeneity, greedy algorithm, asymptotic optimality.
\end{IEEEkeywords}

\section{Introduction}

Controlling the false discovery rate (FDR) has been a widely used approach for multiple testing problems, since the seminal work by Benjamini and Hochberg~\cite{benjamini1995}. 
The Benjamini-Hochberg (BH) procedure was initially proved to control the FDR under the assumption of independent test statistics; since then, there has been extensive research on the theoretical extensions of the BH procedure~\cite{benjamini2001control,storey2002direct,sarkar2002some,genovese2002operating,blanchard2008,ramdas2019unified, efron2012large} and its numerous applications~\cite{genovese2002thresholding,abramovich2006adapting,efron2012large,golz2022multiple}.

 Different from the classical distributed detection formulations~\cite{tenney1981,tsitsiklis1984,viswanathan1997,blum1997} for a \emph{global} null hypothesis, this work concerns the distributed inference problem under the \emph{FDR control}, where each node processes a large number of local hypotheses. For sensor network applications, by assuming a \emph{broadcast model} where each sensor can broadcast its decision to the entire network, several (iterative) distributed BH procedures have first been studied in~\cite{ermis2005adaptive,ray2007novel,ermis2009distributed,ray2011false}. Specifically, given that each sensor has only one p-value (as in~\cite{ermis2005adaptive,ray2007novel,ray2011false}) or needs to transform its local p-values to a scalar (as in~\cite{ermis2009distributed}), the proposed procedures only need each sensor to broadcast at most $1$-bit information. This broadcast model, however, is mainly suitable for small-scale networks (e.g., see~\cite[Section VIII]{ermis2009distributed} for discussions). More recently, the authors in~\cite{Ramdas2017b} have developed the QuTE algorithm, a distributed BH procedure for multi-hop networks with general graph structures, which requires each node to transmit its local p-values to all of its neighbor nodes. 
 A quantized version of QuTE has been developed in~\cite{xiang2019distributed}, and~\cite{Ramdas2022} shows that $m$ levels of quantization for {a total of $m$ p-values in the network} is sufficient for making the same rejections as the {QuTE procedure~\cite{Ramdas2017b}}.
The sample-and-forward method in~\cite{pournaderi2022sample} shows that a sampling operation at each node can reduce the communication cost of the QuTE algorithm to $\Oc(\log m)$ bits per node while preserving the provable FDR control and competitive statistical power. 

In this work, we take an \emph{asymptotic} perspective, as the total number of p-values $m$ approaches infinity, by focusing on star networks with a fixed number of nodes, namely, one center node and $N$ local nodes. 
To shed light on the finite-sample regime, we provide a detailed analysis of the rate of convergence for both FDR and power. The asymptotic regime has been used to develop methods beyond the traditional threshold procedures (including the BH procedure), which rejects the p-values ($\Pv = (P_1,..., P_m)$) that are less than a certain threshold $\tau(\Pv)$. It has been shown that a more general class of procedures called the interval procedures, which rejects the p-values that lie in one or multiple intervals of the form $[\tau_l(\Pv), \tau_u(\Pv)]$, can outperform the power of the BH procedure (such as the multiple reference points method~\cite{chi2007performance} and the scan procedure~\cite{ARIASCASTRO202142}). In the distributed setting and under fairly general conditions, all of these methods control the global FDR asymptotically when they are {performed locally in a network}. On the other hand, these methods can be applied to the pooled p-values (as in the centralized setting) in the network and result in a different decision rule and asymptotic performance. We refer to this centralized performance as the \emph{pooled performance}. One natural question is whether one can achieve the pooled performance of these methods in a communication-efficient manner. To address this question, we simplify the setting by assuming that the local distribution of the alternatives is fixed throughout the nodes and attempt to compensate for the mismatch between the global and local proportions of the true null hypotheses. This formulation leads to our \emph{proportion-matching} method,
directly motivated by a seminal approach via mixture models~\cite{genovese2002operating}, which provides a characterization of the cutt-off threshold as the number of p-values approaches infinity. The proportion-matching method only requires a one-shot communication of $2 \lceil\log_2 m^{(i)} \rceil$ bits from each node with $m^{(i)}$ denoting the number of p-values at node $i$,
and achieves the pooled performance of the BH procedure.
Even though our asymptotic analysis requires some technical assumptions, the proportion matching method proves stable and robust performance  
both theoretically and in various challenging numerical settings.

Allowing for more sophistication, we provide an explicit characterization of the asymptotically optimal solution for the distributed inference problem and show that it can be approximated by our \emph{greedy aggregation} method which is a communication- and computation-efficient iterative algorithm. 
This method assigns a unique rejection region to each node in the network, and
 has asymptotically optimal power for a class of procedures designed for the distributed setting (see Section~\ref{sec:optimal}). The method relies on one key idea: The center node greedily picks the intervals of the highest p-value ``density" among all the nodes in multiple rounds until the FDR constraint is met; the communication cost is $\mathcal{O}(\log m)$ bits (see Section~\ref{sec:optimal}-B).

 In this work, the p-values are modeled as \iid random variables following a mixture distribution for the sake of simplicity. However, we remark that we only rely on the fundamental convergence results that hold for a much broader class of sequences than \iid, {namely, ergodic sequences, which include stationary Gaussian processes with vanishing correlations, stationary $\mathsf{ARMA}$ models, 
 $\mathsf{m}$-dependent stationary sequences, and ergodic Markov chains.} Thus, we expect our proposed methods to be applicable to a wide range of practical problems including massive data collection in wireless sensor networks (WSNs)~\cite{djedouboum2018big}. For instance, in the case of large-scale WSNs, it is reasonable and useful to adopt a hierarchical architecture that partitions the sensors into clusters~\cite{ari2016power,bajaber2014efficient}. In each cluster, a leader called cluster head (CH) coordinates the local sensors and interfaces with the rest of the network (fusion center and/or other CHs depending on the between-cluster topology). In such a setting, CHs should handle a relatively large number of observations. Considering CHs as local nodes in our algorithms, the aforementioned architecture is compatible with our proposed methods regardless of the topology adopted within and between clusters, while the star topology in our work is assumed for simplicity of presentation. Applications of such networks include environmental monitoring~\cite{fascista2022toward} and smart cities~\cite{khalifeh2021wireless}. Multimedia WSNs are another example of massive data collection in networks where communication efficiency is a natural concern as nodes monitor sound, images, and videos~\cite{misra2008survey}.
 It is also noteworthy that our greedy aggregation method can be applied to centralized settings where side information including spatial or temporal information can be leveraged to partition the p-values. For instance, for the region of interest (ROI) analysis~\cite{rosenblatt2018all}, one can partition the image into segments and treat each region as a node with multiple p-values (corresponding to voxel-wise p-values), and our algorithm can potentially lead to more efficient and more accurate localization.

The main contributions of this work is fourfold. First, we present the proportion-matching algorithm which requires transmitting $2 \lceil\log_2 m^{(i)} \rceil$ bits from each node $1\leq i \leq N$, yet it achieves the global BH performance asymptotically; {we characterize the rate of convergence for both the FDR and power}. Second, we analyze its robustness against heterogeneity in alternative distributions as well as the inconsistency of estimators of the proportion of nulls at different nodes. {Third, we provide an explicit characterization of the asymptotically optimal rejection regions, which serves as a baseline given known probability density functions and can lead to an efficient algorithm when coupled with density estimation methods. Fourth, we develop the greedy aggregation algorithm and show that it is asymptotically optimal in power (in a particular sense suitable for distributed settings), {along with the analysis of the rate of convergence for FDR and power}; its empirical performance is superior over a wide range of simulation settings.} Preliminary results on the proportion-matching method have been reported in~\cite{pournaderi2022communication} by focusing on homogeneous alternatives and consistent estimators of the ratio of nulls (see Assumption~\ref{assume: consistent}), and we have extended this to allow for heterogeneous alternatives and estimators that converge to an upper bound of the proportion of nulls, which is more in line with the existing estimators listed in Section~\ref{ssec:prop_est}.

The rest of the paper is organized as follows. We present the background and problem formulation in Section~\ref{sec:background}. In Section~\ref{sec:comm-eff BH}, we show the asymptotic analysis of the proportion-matching method and analyze its robustness in heterogeneous settings in Section~\ref{sec:robust}. In Section~\ref{sec:optimal}, we develop the greedy aggregation method along with a characterization of the asymptotic optimal rejection regions. Numerical simulations are presented in Section~\ref{sec:sim} to support our theoretical findings. 

\vspace{-.5em}
\section{Background}
\label{sec:background}

\subsection{Multiple Testing and False Discovery Rate Control}
\label{sec:FDR}

Let $(X_1,...,X_m)$ denote $m$ statistics computed according to some observed data and consider testing the null hypotheses $\mathsf{H}_{0,k}:X_k\sim\mathsf{F_{0,k}}, 1 \leq k\leq m$ where $\mathsf{F_{0,k}}$ denotes the (hypothetical) CDF of $X_k$ under $\mathsf{H}_{0,k}$. When $\mathsf{F_{0,k}}$ is continuous, by the probability integral transform, the tests can equivalently be written as $\mathsf{H}_{0,k}:P_k\sim\mathsf{Unif}(0,1), 1 \leq k\leq m$ with $P_k = \mathsf{Q}_{{0,k}}(X_k)$ denoting the p-values and $\mathsf{Q}_{{0,k}}=1-\mathsf{F}_{{0,k}}$. Let $m_0$ denote the number of statistics for which the null hypotheses are true, i.e., their corresponding p-values have a uniform distribution on $[0,1]$. We refer to these $m_0$ p-values as \emph{null p-values} and the remaining $m_1 = m-m_0$ as \emph{non-nulls}. 

The purpose of multiple testing is to test the $m$ hypotheses while controlling a simultaneous measure of type I error. A rejection procedure controls a measure of error at some (prefixed) level $\alpha$ if it guarantees to keep the error less than or equal to $\alpha$. Two major approaches to a simultaneous measure of error are the family-wise error rate (FWER) control and the false discovery rate (FDR) control. FWER is concerned with controlling the probability of making at least one false rejection, while FDR is a less stringent measure of error that aims to control the average \emph{proportion} of false rejections among all rejections. Let $R$ and $V$ denote the number of rejections and false rejections by some procedure, respectively. Then, FWER is defined as $\mathbb{P}(V>0)$ and it can be controlled at some target level $\alpha$ by rejecting $\Rc_{\text{Bonf}}=\{k: P_k \leq \alpha/m\}$. The threshold $\alpha/m$ is so-called Bonferroni-corrected~\cite{bonferroni1936teoria} test size. It is known that using the Bonferroni correction results in too conservative rejection rules when $m$ is large. On the other hand, FDR is defined as the expected value of the false discovery proportion (FDP), i.e.,  $\mathsf{FDR}=\mathbb{E}(\mathsf{FDP})$ with $\mathsf{FDP} = \frac{V}{R\vee 1}$,
($a\vee b:=\max\{a,b\}$) and allows for more rejection. The celebrated Benjamini-Hochberg (BH) procedure~\cite{benjamini1995} controls the $\mathsf{FDR}$ at level $\a$. Let $(P_{(1)},P_{(2)},\ldots, P_{(m)})$ denote the ascending-ordered p-values. The BH procedure rejects $\hat{k}$ smallest p-values, i.e., $\Rc_{\text{BH}}=\{k: P_k \leq P_{(\hat{k})}\}$, where 
\begin{equation*}
    \kh = \max\big\{1\leq k\leq m: P_{(k)}\le \tau_k \big\},
\end{equation*}
with $\tau_k=\a k/m$ 
and $\hat{k}=0$ if $P_{(k)} > \tau_k$ for all $1\leq k\leq m$. We adopt notations $D_{\text{BH}}:=\hat k$ and $\tau_{\text{BH}}:= \tau_{\hat{k}}=\alpha D_{\text{BH}}/m$ for the deciding index and rejection threshold respectively. 

The other performance measure of a multiple testing procedure is the power of detection, defined as the expectation of the true discovery proportion (TDP), i.e., $\mathsf{power}=\mathbb{E}(\mathsf{TDP})$ with $\mathsf{TDP}={\frac{R-V}{m_1\vee 1}}$,
where we follow the definition from~\cite{benjamini1995}.

\vspace{-1em}
\subsection{Distributed False Discovery Control: Problem Setting}
\label{sec:distributed}
In the distributed multiple testing problem setting, we aim to develop statistically powerful algorithms that are communication- and
computation-efficient under the asymptotic FDR control constraint. 

Consider a network consisting of one center node and $N$ other nodes (and $N$ is \emph{fixed} throughout the work), where the center node collects (or broadcasts) information from (or to) all the other $N$ nodes. {We follow the widely used random effect (or hierarchical) model in the large-scale inference literature (e.g.,~\cite{efron2001empirical,storey2002direct,genovese2004stochastic}) and assume} the p-values in the network are generated $\iid$ according to the (mixture) distribution function
$\Gc(t)=\sum_{i=1}^N{q^{(i)}\,G_i(t)}$, where a p-value is generated at node $i$ with probability $q^{(i)}>0$ and according to the distribution function {$G_i(t) = r_0^{(i)}\,U(t) + r_1^{(i)} F_i(t),\ 0 < r_0^{(i)} < 1,\ r_1^{(i)}:=1-r_0^{(i)}$},
with $U$ denoting the CDF of p-values given their corresponding null hypothesis is true, (which is $\mathsf{Unif}(0, 1)$,) and $F_i$ denoting the \emph{unknown} alternative distribution function of the p-values at node $i$. {Both methods presented in this paper are fully data-driven and no distributional parameters are assumed to be known at the center node.} Let $m$ denote the total number of p-values in the network where each node $i$, $i\in \{1,...,N\}$, owns $m^{(i)}$ 
p-values $\Pv^{(i)}=(P_1^{(i)},..., P_{m^{(i)}}^{(i)})$, with $m^{(i)}_0$ (or $m^{(i)}_1$) of them corresponding to the true null (or alternative) hypotheses ($m^{(i)}_0 + m^{(i)}_1 = m^{(i)}$); let $m_0=\sum_i m^{(i)}_0$ and $m_1=m-m_0$. Denote the marginal probability of true null in the network (or asymptotic global proportion of true nulls) as $r_0^*=\sum_{i=1}^N q^{(i)}r_0^{(i)}$, and accordingly $r_1^*=1-r_0^*$. {Throughout the paper, we assume that $q^{(i)},\ r_0^{(i)}$ and $F_i$ are fixed as $m\to\infty$ and we drop the dependency of FDR and power on $m$ for simplicity of presentation.}

\vspace{-.5em}
\section{Communication-Efficient Algorithms}
\label{sec:comm-eff BH}
 
\subsection{{Local Multiple Testing (No Communication)}}
In this section, we study the global performance of multiple testing procedures when they are applied locally in a network and no communication is permitted. A general statement is made in Proposition~\ref{lem:no-com}, and Proposition~\ref{thm:no-com-BH} particularly focuses on the local BH procedures under the setting described in \ref{sec:distributed}.
\begin{propo}\label{lem:no-com}
{Let $R_m^{(i)}$, $V_m^{(i)}$, and $\mathsf{FDP}^{(i)}$ denote the number of rejections, false rejections, and FDP of an arbitrary procedure at note $i$, respectively.} If $\underset{m\rightarrow\infty}{\overline{\lim}}\mathsf{FDP}^{(i)}\leq \alpha\ 
a.s.$, $1\leq i \leq N$, then $\underset{m\rightarrow\infty}{\overline{\lim}}\mathsf{FDR}\leq \alpha$.
\end{propo}
\begin{proof}
Note that the global FDR is $\mathbb{E}\left(\frac{\sum_{i=1}^N V_m^{(i)}}{1\vee \sum_{i=1}^N R_m^{(i)}}\right)$. Thus,
\begin{align*}
\underset{m\rightarrow\infty}{\overline{\lim}}\mathsf{FDR}
    &\leq \underset{m\rightarrow\infty}{\overline{\lim}}\mathbb{E}\biggl(\underset{1\leq i\leq N}{\max}\frac{ V_m^{(i)}}{1\vee  R_m^{(i)}}\biggr)\\
    &\overset{(a)}{\leq} \mathbb{E}\biggl(\underset{m\rightarrow\infty}{\overline{\lim}}\underset{1\leq i\leq N}{\max}\frac{ V_m^{(i)}}{1\vee  R_m^{(i)}}\biggr)\\
    &= \mathbb{E}\biggl(\underset{1\leq i\leq N}{\max}\underset{m\rightarrow\infty}{\overline{\lim}}\mathsf{FDP}^{(i)}\biggr)\leq \alpha,
\end{align*}
where $(a)$ follows from Fatou's lemma.
\end{proof}

{We now wish to focus on the BH procedure specifically.} Let $\tau^{(i)}_{\text{BH}}(\alpha)$ denote the local BH rejection thresholds at node $i$ for test size $\alpha$.
Under suitable assumptions, \cite[Theorem~$1$]{genovese2002operating} argues the following {(see Appendix~\ref{techlem} for a stronger version)},
\begin{align}
    \tau^{(i)}_{\text{BH}}(\alpha)&\xrightarrow{\mathcal{P}} %\tau(\alpha^{(i)};r_0^{(i)}),
    \sup\big\{t:G_i(t)=t/\alpha\big\}=:\theta^{(i)}, \label{eq:asy-loc-th}
\end{align} 
{where the existence of $\theta^{(i)}$ is guaranteed by the Knaster–Tarski theorem~\cite{tarski1955lattice}. The following proposition concerns the asymptotic FDR analysis of the BH procedure when applied locally to the distributed setting.}
\begin{propo}\label{thm:no-com-BH}
Assume $G_i(t)$ is continuously differentiable at $t=\theta^{(i)}$ and $G_i'\big(\theta^{(i)}\big)\neq  1/\alpha$ for all $1\leq i\leq N$. {Performing local BH procedures controls the global FDR asymptotically with $\mathsf{FDR}\leq \alpha + o(\log m/\sqrt{m})$, as long as $G_i(\theta^{(i)})>0$ for at least one node $1\leq i\leq N$.}
\end{propo}
\begin{proof}
Denote the number of rejections by $R_m^{(i)}(\tau^{(i)}_{\text{BH}})=\sum_{k=1}^{m^{(i)}}{\ind\{P_k^{(i)}\leq \tau^{(i)}_{\text{BH}}\}}$ and the number of false rejections by $V_m^{(i)}(\tau^{(i)}_{\text{BH}})=\sum_{k=1}^{m^{(i)}}{\ind\{\mathsf{H}_{0,k}^{(i)}=1,P_k^{(i)}\leq \tau^{(i)}_{\text{BH}}\}}$,
where $\mathsf{H}_{0,k}^{(i)}=1$ denotes that the $k$-th null hypothesis at node $i$ is true, and similarly for $\mathsf{H}_{0,k}^{(i)}=0$.
The weak convergence of empirical processes yields
\begin{equation*}
    \left|\frac{1}{m} \sum_{i=1}^N V_m^{(i)}(\tau^{(i)}_{\text{BH}})-\sum_{i=1}^{N}q^{(i)}r_0^{(i)}\tau^{(i)}_{\text{BH}}\right|= \Oc_p\left(m^{-1/2}\right),
\end{equation*}
\begin{equation*}
    \left|\frac{1}{m} \sum_{i=1}^N R_m^{(i)}(\tau^{(i)}_{\text{BH}})-\sum_{i=1}^{N}q^{(i)}G_i\left(\tau^{(i)}_{\text{BH}}\right)\right|= \Oc_p\left(m^{-1/2}\right),
\end{equation*}
Also, according to Lemma~\ref{het} (with $\gamma=1/2$) and Taylor's theorem we have
\begin{equation*}
    \left|\sum_{i=1}^{N}q^{(i)}r_0^{(i)}\tau^{(i)}_{\text{BH}}-\sum_{i=1}^{N}q^{(i)}r_0^{(i)}\theta^{(i)}\right|= o_p\left(\log m/\sqrt{m}\right),
\end{equation*}
\begin{equation*}
    \left|\sum_{i=1}^{N}q^{(i)}G_i\left(\tau^{(i)}_{\text{BH}}\right)-\sum_{i=1}^{N}q^{(i)}G_i\left(\theta^{(i)}\right)\right|= o_p\left(\log m/\sqrt{m}\right),
\end{equation*}
Therefore, 
\begin{align*}
    \mathsf{FDP}=\frac{\frac{1}{m}\sum_{i=1}^N V_m^{(i)}(\tau^{(i)}_{\text{BH}})}{\frac{1}{m}\left(1\vee \sum_{i=1}^N R_m^{(i)}(\tau^{(i)}_{\text{BH}})\right)}=&\frac{\sum_{i=1}^{N}q^{(i)}r_0^{(i)}\theta^{(i)}}{\sum_{i=1}^{N}q^{(i)}G_i\left(\theta^{(i)}\right)}\\
    & + o_p\left(\log m/\sqrt{m}\right),
\end{align*} 
if $G_i(\theta^{(i)})>0$ for at least one node. Hence,
\begin{align*}
\mathsf{FDP}&\xrightarrow{\Pc}\frac{\sum_{i=1}^{N}q^{(i)}r_0^{(i)}\theta^{(i)}}{\sum_{i=1}^{N}q^{(i)}G_i\left(\theta^{(i)}\right)}\\
&\qquad    \leq \underset{i:G_i(\theta^{(i)})>0}{\max}\left\{r_0^{(i)}\frac{\theta^{(i)}}{G_i\left(\theta^{(i)}\right)}\right\}\leq \alpha,
\end{align*}
which implies $\mathsf{FDR}\leq \alpha + o(\log m/\sqrt{m})$ according to the boundedness of FDP.
\end{proof}
\begin{remark}
{Similar to  Proposition~\ref{thm:no-com-BH}, one can show that 
\begin{align*}
    \mathsf{power} &= \frac{1}{r_1^*}\sum_{i=1}^Nq^{(i)}r_1^{(i)} F_i\left(\theta^{(i)}\right)+o(\log m/\sqrt{m})\\
    &= \frac{1}{r_1^*}\sum_{i=1}^Nq^{(i)}r_1^{(i)} \mathsf{power}^{(i)}+o(\log m/\sqrt{m}),
\end{align*}
which is the weighted average of local powers.}
\end{remark}

\vspace{-1em}
\subsection{Proportion-Matching Algorithm}\label{sec:prp-m}
In this section, we present an algorithm to correct for the mismatch between the local proportion of nulls and global one, to achieve the global BH performance when the alternative distribution is fixed throughout the network.
\begin{assumption}
We assume the following (and we will relax this assumption in Section~\ref{sec:robust})
\begin{equation*}
    G_i(t)=G\big(t;r_0^{(i)}\big) := r_0^{(i)}\,U(t) + (1-r_0^{(i)}) F(t),
\end{equation*}
i.e., the local (conditional) distributions differ only in proportion of the nulls and $F_i=F$ is fixed for all $1\leq i\leq N$.\label{ass:fixed_dist}
\end{assumption}
Let $\tau^*_{\text{BH}}(\alpha)$ and $\tau^{(i)}_{\text{BH}}(\alpha^{(i)})$ denote the global and local ($i$-th node) rejection thresholds for the test sizes $\alpha$ and $\alpha^{(i)}$, respectively. Similar to \eqref{eq:asy-loc-th}, we have
\begin{align*}
    \tau^*_{\text{BH}}(\alpha)&\xrightarrow{\mathcal{P}}%\sup\big\{t:G(t;r_0^*)=(1/\alpha)\,t\big\}
    \tau(\alpha;r_0^*),\\
    \tau^{(i)}_{\text{BH}}(\alpha^{(i)})&\xrightarrow{\mathcal{P}} \tau(\alpha^{(i)};r_0^{(i)}),
    %\sup\big\{t:G\big(t;r_0^{(i)}\big)=\big(1/\alpha^{(i)}\big)\,t\big\}.
\end{align*} 
where $\tau(\alpha;r_0):=\sup\big\{t:G(t;r_0)=(1/\alpha)\,t\big\}$.
On the other hand, we have
\begin{equation*}
    \sup\{t:G(t;r_0)=(1/\alpha)\,t\}=\sup\{t:F(t)=\beta(\alpha;r_0)\,t\}
\end{equation*}
for $0 \leq r_0<1$, where $\beta(\alpha;r_0) := \frac{(1/\a)-{r}_0}{1-{r}_0}$. 
Therefore, $\tau^*_{\text{BH}}$ admits a characterization via $F(t)$ and $\beta(\alpha;r_0^*)$ in the limit (and same for $\tau^{(i)}_{\text{BH}}$)~\cite{genovese2002operating}.
This (asymptotic) representation leads to a key observation: Each node can leverage the global proportion of true nulls $r_0^*$ (provided to them) to achieve the global performance by calibrating its (local) test size $\alpha^{(i)}$ such that $\beta(\alpha^{(i)};r_0^{(i)})=\beta(\alpha;r_0^*)=:\beta^*$ to reach the global threshold. Specifically, it is straightforward to observe that setting $\a^{(i)} = \big((1-{r}_0^{(i)})\beta(\alpha;r_0^*) + {r}_0^{(i)}\big)^{-1}$ at node $i$ will result in the global performance via the limiting threshold $\tau(\alpha;r_0^*)$. We now formally present our proportion-matching algorithm.

Let $\hat{r}^{(i)}_0$ denote an estimator of $r^{(i)}_0$ (see the next subsection for estimation procedures). 
For an overall targeted FDR level $\alpha$, our \emph{proportion-matching  method} consists of three steps. 
\begin{itemize}\setlength\itemsep{0.5em}
	\item[(1)] {\bf Collect p-values counts}: Each node~$i$ estimates $\hat{r}^{(i)}_0$ and then sends $\big(m^{(i)}, \hat{r}^{(i)}_0\big)$ to the center node.
	\item[(2)] {\bf Estimate global slope ($\b^*$)}: Based on $\big(m^{(i)}, \hat{r}^{(i)}_0\big)$, $i\in\{1,...,N\}$, the center node computes $\hat\beta^*$ and broadcasts it to all the nodes, where
	\begin{align*}
		\hat\beta^* = \frac{(1/\a)-\hat{r}^*_0}{1-\hat{r}^*_0}\geq 1,
	\end{align*}
	 with $\hat{r}^*_0= \frac{1}{m}\sum_{i=1}^N \hat{r}_0^{(i)}m^{(i)}$ and $m=\sum_{i=1}^N m^{(i)}$.
	\item[(3)] {\bf Perform BH locally}: Upon receiving $\hat\b^*$, each node computes its own $\hat\a^{(i)}$ and performs the BH procedure according to $\hat\a^{(i)}$, where
	\begin{align*}
		\hat\a^{(i)} = \frac{1}{\big(1-\hat{r}_0^{(i)}\big)\hat\b^* + \hat{r}_0^{(i)}}\leq 1.
	\end{align*}
\end{itemize}

\begin{remark}
We observe that in the case of having only one node in the network ($N=1$), we get $\hat{r}^*_0= \hat{r}_0^{(1)}$ and the algorithm reduces to the BH procedure, i.e., $\hat\a^{(1)}=\alpha$.
\end{remark}
 
We note that in step (1) each node can send $\hat{m}^{(i)}_0=\lfloor \hat{r}_0^{(i)} m^{(i)} + 1/2\rfloor$ instead of $\hat{r}_0^{(i)}$ and in step (2) the center node can broadcast $(m,\sum_{i=1}^N \hat{m}^{(i)}_0)$ instead of $\hat\beta^*$. In this case, {assuming all nodes are aware of the global target FDR $\alpha$}, each node can estimate $\hat\beta^*$ using $\hat{\pi}^*_0=\frac{1}{m}\sum_{i=1}^N \hat{m}^{(i)}_0$ as an estimator of $r_0^*$.
Thus, each node only needs to transmit $2\lceil\log_2 m^{(i)} \rceil$ bits, and the center node broadcasts $2\lceil\log_2 m \rceil$ bits. This is in contrast with communicating quantized p-values which requires $\Oc(m\log(m))$ bits of communication~\cite{Ramdas2022} to achieve the global BH performance. 

\vspace{-1em} 
\subsection{Estimation of $r_0^{(i)}$}
\label{ssec:prop_est}
In this section, we briefly review two existing estimators of $r_0^{(i)}$. Let $P_{(1)}<...< P_{(m)}$ denote the ordered p-values at some node, generated \iid from the distribution function $G(t)$.
The problem of estimating the ratio of alternatives $r_1$ has been studied extensively~\cite{hochberg1990more,hengartner1995finite,swanepoel1999limiting,benjamini2000adaptive,efron2001empirical,storey2002direct}. Two estimators of $r_0$ are given below: 

\smallskip
\noindent{{\bf Spacing estimator}}~\cite{swanepoel1999limiting}: 
Take a sequence $\{s_m\}$ such that $s_m/m\to 0$ and $s_m/\log(m)\to \infty$. Let
\begin{equation*}
		\hat{r}^{\text{spacing}}_0 = \min\biggl\{\frac{2 s_m}{m Z_m},1\biggr\},
	\end{equation*}
where $Z_m = \max_{s_m+1\le j\le m-s_m} (P_{(j+s_m)}-P_{(j-s_m)})$.
 
Assume that $G(t)$ is differentiable with density function $g(t)=G'(t)$; note that $g(t)=r_0 + r_1 f(t)$ where $f(t)$ denotes the probability density function of the alternative distribution $F(t)$. Under mild assumptions on $g(t)$, \cite{swanepoel1999limiting} proves that $\hat{r}^{\text{spacing}}_0\xrightarrow{a.s.} \inf_{0<t<1} g(t)\geq r_0$. Although in general the limiting value of $\hat{r}^{\text{spacing}}_0$ is an upper bound of the parameter, it should be noted that this estimator is (strongly) consistent for one-sided p-values computed according to the continuous exponential families, i.e., $\inf_{0<t<1} f(t)=0$ and hence $\hat{r}^{\text{spacing}}_0\xrightarrow{a.s.} r_0$ for these distributions. Also, the gap between the limiting value and $r_0$ will be quite small in practical cases where the estimator is not consistent, e.g., two-sided p-values for testing the location of normal and t-distribution \cite{genovese2004stochastic}. The following lemma is a corollary to Theorem 1.1 in \cite{swanepoel1999limiting}.
\begin{lemma}
   % Let $s_m \ge u_m$ for all $m \geq 1$ and $s_m/m\to 0$.
   Let $\overline{r}_0^{\text{spacing}} := \inf_{0<t<1} g(t)$. Assume that $g$ attains its infimum at some unique {$\theta\in(0,1)$} and  $g'(\theta)<\infty$ 
   %(with $g$ and $g'$ denoting the right or left derivative for boundary cases)
   , then
\begin{align*}
    \Big|\hat{r}^{\text{spacing}}_0 &-\overline{r}_0^{\text{spacing}}\Big| \\=
    & \Oc\left\{\max\left(s_m^{-1/2}\log\left(\frac{m}{s_m}\right),\frac{s_m}{m}+\sqrt{\frac{\log\log m}{m}}\right)\right\}\\
    = & \Oc\left\{\max\left(s_m^{-1/2}\log\left(\frac{m}{s_m}\right),\frac{s_m}{m}\right)\right\}\quad a.s.
\end{align*}   
\end{lemma}
\begin{proof}
    The proof follows from Equations~(2.1) to (2.3) in the proof of Theorem~1.1 in \cite{swanepoel1999limiting} (with $r_n=s_n$), {Taylor's theorem,} and Equation~(4) from \cite{bahadur1966note}.
\end{proof}
For example if $s_m = m^{3/4}$, then $\hat{r}^{\text{spacing}}_0=\overline{r}_0^{\text{spacing}}+\Oc\left(m^{-1/4}\right)$ almost surely.

\smallskip
\noindent{{\bf Storey's estimator}}~\cite{storey2002direct}: Let $\widehat{G}(t)$ denote the empirical CDF of $(P_1,..., P_m)$. For any $\l\in (0,1)$, 
	\begin{align*}
		\hat{r}^{\text{Storey}}_0 = \min\biggl\{\frac{1-\widehat{G}(\l)}{1-\l}, \,1\biggr\};
	\end{align*} 
This estimator converges almost surely to an upper bound of $r_0$. 
This follows from some elementary algebra combined with $\widehat{G}(\l)\xrightarrow{a.s.} G(\lambda)$ by the strong law of large numbers and the convergence rate is  
$\Oc_p\left(m^{-1/2}\right)$.

\vspace{-1em}
\subsection{Asymptotic Equivalence to the Global BH}

To facilitate the presentation, we list some key notions.
\begin{center}
\begin{tabular}{ c|c|c|c|c } 
%\hline
  & ratio of nulls &level & slope \\
\hline
\multirow{3}{4em}{parameter \\ estimate \\ limit} & $r_0^*=1-r_1^*$ &$\a^{(i)}$  &$\b^*=\beta(\alpha;r_0^*)$\\ 
& $\hat{r}^*_0$ & $\hat{\a}^{(i)}$& $\hat{\b}^*$  \\ 
& $\overline{r}^*_0$ (upper) & $\tilde{\a}^{(i)}$ & $\overline{\b}^*$ (upper)\\  
\hline
\end{tabular}
%\caption{Table for notation}
%\label{table:1}
\end{center}

\begin{center}
\begin{tabular}{ c|c|c|c|c } 
%\hline
  & rejection threshold  \\
\hline
\multirow{2}{2em}{global \\ [1.5ex]local } & $\tau_{\text{BH}}^*=\tau_{\text{BH}}(\alpha)$, $\tau^*=\tau(\alpha,;r_0^*)$ (asymptotic) \\ [1ex]
& $\tau_{\text{BH}}^{(i)}(\alpha^{(i)})$, $\tau_{\text{BH}}^{(i)}(\hat{\alpha}^{(i)})$, $\tau_{\text{BH}}^{(i)}=\tau_{\text{BH}}^{(i)}(\alpha)$ (no comm.)  \\
\hline
\end{tabular}
%\caption{Table for notation}
%\label{table:1}
\end{center}

\begin{lemma}\label{lem:lln}
If $\hat{r}_0^{(i)}\xrightarrow{a.s.}\overline{r}_0^{\,(i)}$ as $m^{(i)}\rightarrow\infty$ for all $i$, then
\begin{equation}
    \hat{r}^*_0\xrightarrow{a.s.}\overline{r}^*_0\ ,\quad
    \hat\beta^*\xrightarrow{a.s.}\overline{\beta}^*\ ,\quad
    \hat\a^{(i)}\xrightarrow{a.s.}\widetilde{\alpha}^{(i)}\ ,\label{converg}
\end{equation}
as $m\to \infty$, where
\begin{subequations}
\begin{align}
   &\overline{r}^*_0 =\sum_{i=1}^{N}{\overline{r}_0^{\,(i)}q^{(i)}}\ ,\;\;\;\;\;\;
    \overline{\beta}^* = \frac{({1}/{\alpha})-\overline{r}^*_0}{1-\overline{r}^*_0} \ ,\nonumber\\
    &\widetilde{\alpha}^{(i)} = \Big({\big(1-\overline{r}_0^{\,(i)}\big)\overline{\beta}^*+\overline{r}_0^{\,(i)}}\Big)^{-1}\ .\nonumber
\end{align}
\end{subequations}
Furthermore, assuming $\left|\hat{r}_0^{(i)}-\overline{r}_0^{\,(i)}\right|=\Oc_p\left(u^{(i)}(m)\right)$,  
the convergence {rates in~\eqref{converg} are} $\Oc_p\left\{\max\left(\underset{i}{\max}\ u^{(i)}(m),m^{-1/2}\right)\right\}$.
\end{lemma}
\begin{proof}
{Since $q^{(i)}_m=q^{(i)}>0$ is fixed for each node, we get $\sum_{m=1}^\infty{q^{(i)}_m}=\infty$ and hence $m^{(i)}\rightarrow\infty$ with probability $1$ for all $i$ as $m\rightarrow\infty$ according to the Borel-Cantelli lemma. We note that $m^{(i)}/m\xrightarrow{a.s.}q^{(i)}$ by Kolmogorov’s strong law of large numbers. The results follow from Taylor's theorem.}
\end{proof}
{Note that if the Storey's estimator is adopted at all local nodes, the convergence rate {of $\hat\a^{(i)}$ would be} $\Oc_p\left(m^{-1/2}\right)$. However, if an estimator with a slower rate is adopted at one node, the convergence rate of {$\hat\a^{(i)}$} at all the other nodes will be impacted.}

\begin{assumption}
\label{assume: consistent}
Estimators of $r_0^{(i)}$ are consistent, i.e., $\overline{r}_0^{\,(i)}=r_0^{\,(i)}$ for all $1\leq i\leq N$.  \label{ass:consis}
\end{assumption}
We note that under Assumption \ref{ass:consis}, we have $\overline{r}_0^*=r_0^*$, $\overline{\beta}^*=\beta^*=\beta(\alpha;r_0^*)$, and $\widetilde{\alpha}^{(i)}={\alpha}^{(i)}$.
Recall that $F$ denotes the (common) alternative CDF of the p-values and define
\begin{equation*}
    \tau^*:=\tau(\alpha;r_0^*)=\sup\{t:F(t)={\beta}^*\, t\}.
\end{equation*}

\begin{assumption}
$F(t)$ is continuously differentiable {in a neighborhood of} $t={\tau}^*$ and $F'({\tau}^*)\neq {\beta}^*$. 
\label{concav}
\end{assumption}
Recall that $\beta(\alpha^{(i)};r_0^{(i)})=\beta^*$ and $\tau(\alpha^{(i)};r_0^{(i)})=\tau^*$ for $\a^{(i)} = \big((1-{r}_0^{(i)})\beta(\alpha;r_0^*) + {r}_0^{(i)}\big)^{-1}$ as discussed in Section~\ref{sec:prp-m}. Therefore, according to Lemma \ref{het} (in Appendix \ref{techlem}), if $\alpha^{(i)}$'s are known, then Assumption~\ref{concav} is sufficient to imply $\tau_{\text{BH}}^{(i)}({\alpha}^{(i)})\xrightarrow{\Pc} \tau^*$ as $m\to\infty$, where $\tau_{\text{BH}}^{(i)}({\alpha}^{(i)})$ denotes the BH rejection threshold at node $i$ with the target FDR ${\alpha}^{(i)}$.  
The following theorem concerns the asymptotic validity of this argument for the BH procedure based on the plug-in estimator $\hat{\alpha}^{(i)}$, i.e., we wish to show that $\tau_{\text{BH}}^{(i)}(\hat{\alpha}^{(i)})\xrightarrow{\Pc} {\tau}^*$ holds as well. 
\begin{theorem}
Assume $m$ p-values generated according to the mixture model $r_0 U(t)+(1-r_0) F(t),\ 0\leq r_0 <1$, under Assumptions~\ref{concav} with $\tau^*=\tau(\widetilde\alpha;r_0)$. Let
\begin{equation}
    \tau_{\text{BH}}(\hat\alpha_m) = \frac{\hat\alpha_m}{m}\max\big\{0\leq k\leq m:P_{(k)}\leq (k/m)\hat\alpha_m\big\}\  \nonumber
\end{equation}
(with convention $P_{(0)} = 0$) denote the rejection threshold for the BH procedure with estimated target FDR $\hat\alpha_m$.
If $\left|\hat{\alpha}_m-\widetilde\alpha\right|=\Oc(t_m)\ a.s.$ with $t_m\to 0$ as $m\to\infty$, then $ \tau_{\text{BH}}(\hat\alpha_m)\xrightarrow{a.s.}\tau(\widetilde\alpha;r_0)=:\tau({\widetilde\alpha})$ and $\tau_{\text{BH}}(\hat\alpha_m)=\tau({\widetilde\alpha})+\Oc(t_m\vee m^{-1/2}(\log m)^{3/2})$ almost surely.
\label{thm:plugin}
\end{theorem}
\begin{proof}
According to Assumption~\ref{concav}, there exist a $\delta$-neighborhood of $\widetilde\alpha$, $\mathcal{B}_{\delta}(\widetilde\alpha)$, such that $F'(\tau_{\alpha'})\neq \beta_{\alpha'}$ for all $\alpha'\in\mathcal{B}_{\delta}(\widetilde\alpha)$. 
{Take $t'_m = c\, t_m$ such that $|\hat\alpha_m-\widetilde\alpha|\leq t'_m\ a.s.$ for large $m$.
In this case, for large $m$ we get 
\begin{equation*}
    \tau_{\text{BH}}(\widetilde\alpha-t'_m)\leq \tau_{\text{BH}}(\hat\alpha_m)\leq \tau_{\text{BH}}(\widetilde\alpha+t_m') \quad \ a.s.\ .
\end{equation*}
Take the sequence $\varepsilon_m = m^{-1/2}(\log m)^{3/2}$. Then according to Lemma~\ref{het} (with $\kappa=1/2$), we get 
\begin{equation*}
    (1-\varepsilon_m)\tau({\widetilde\alpha-t_m'})\leq \tau_{\text{BH}}(\hat\alpha_m)\leq (1+\varepsilon_m)\tau({\widetilde\alpha+t_m'}) \ a.s.
\end{equation*}
for large $m$. 
According to Assumption~\ref{concav} and the inverse function theorem~\cite{tao2015analysis}, $\tau(\alpha)$ has a (continuous) derivative in a neighborhood of $\widetilde{\alpha}$. Hence, by Taylor's theorem   
we get $\left|\tau_{\text{BH}}(\hat\alpha_m)-\tau({\widetilde\alpha})\right|=\Oc(t_m\vee m^{-1/2}(\log m)^{3/2})$ a.s.\ .}
\end{proof}

\begin{corollary}
{Assume the same setting as in Theorem~\ref{thm:plugin}. If $\left|\hat{\alpha}_m-\widetilde\alpha\right|=\Oc_p(t_m)$ with $t_m\to 0$ as $m\to\infty$, then $\tau_{\text{BH}}(\hat\alpha_m)=\tau({\widetilde\alpha})+\Oc_p(t_m\vee m^{-1/2}\log m)$.}
\end{corollary}
\begin{proof}
    {The proof follows a similar argument as in the proof of Theorem~\ref{thm:plugin} with $\varepsilon'_m = m^{-1/2}\log m$ (instead of $\varepsilon_m$).}
\end{proof}

\begin{propo}
Under Assumptions~\ref{ass:fixed_dist},~\ref{ass:consis} and~\ref{concav}, if $\tau^*>0$ then the proportion-matching method attains the pooled BH performance (FDR and power) asymptotically. {
Moreover, if $\left|\hat{r}_0^{(i)}-\overline{r}_0^{\,(i)}\right|=\Oc_p\left(u^{(i)}(m)\right)$, 
then $\left|\tau_{\text{BH}}^{(i)}(\hat{\alpha}^{(i)})-\tau^*_{\text{BH}}(\alpha)\right|=\Oc_p\left(v_m\right)$, $\mathsf{FDR}=r_0^*\alpha+\Oc\left(v_m\right)$, and $\mathsf{power}=\mathsf{power}^*+\Oc\left(v_m\right)$, where $v_m=u_m\vee m^{-1/2}\log m$ with $u_m={\max}_i\ u^{(i)}(m)$, and $\mathsf{power}^*$ denotes the pooled power.}
\end{propo}
\begin{proof}
According to Assumption~\ref{ass:consis} and Lemma \ref{lem:lln}, we have $\hat\alpha^{(i)}\xrightarrow{a.s.}\alpha^{(i)}$ for all $1\leq i\leq N$. Therefore, by Theorem~\ref{thm:plugin}, we get $\tau_{\text{BH}}^{(i)}(\hat{\alpha}^{(i)})\xrightarrow{a.s.} \tau(\alpha^{(i)};r_0^{(i)})=\tau^*$. Also, we have $\tau^*_{\text{BH}}(\alpha)\xrightarrow{a.s.}\tau^*$ according to Lemma \ref{het}. Hence, each node rejects according to the global BH threshold asymptotically. Thus, since $\tau^*>0$, FDR and power of the proportion-matching method also converge to the same limiting values as the global BH asymptotic performance. This follows from the convergence and boundedness of $\mathsf{FDP}$ and $\mathsf{TDP}$ similar to the proof of Proposition~\ref{thm:no-com-BH}. {The convergence rate can be obtained by combining the rates from Lemma
~\ref{lem:lln} and {Corollary~\ref{thm:plugin}} and using the same argument as in the proof of Proposition~\ref{thm:no-com-BH}.}
\end{proof}

\begin{remark}
The proportion matching method can potentially be extended to match the pooled performance of the scan and multiple reference points {procedures (\!\!~\cite{ARIASCASTRO202142,chi2007performance})} as well.
\end{remark}

\vspace{-1em}
\section{Robustness Against Inconsistent Estimators and heterogeneous Alternatives}
\label{sec:robust}
In this section we wish to relax Assumptions~\ref{ass:fixed_dist} and \ref{ass:consis} in the proportion matching method. {Specifically, Theorem~\ref{prop_het} concerns relaxing both assumptions simultaneously, while the rest of the section focuses on dropping only Assumption~\ref{ass:fixed_dist}, i.e., a heterogeneous setting is considered, where each node can have a unique alternative distribution function $F_i(t)$.} In this respect, we write $\widetilde{\tau}^{(i)}:=\sup\{t:F_i(t)=\widetilde{\beta}^{(i)}t\}$  
where $\widetilde{\beta}^{(i)}=\beta\big(\widetilde{\alpha}^{(i)};r_0^{(i)}\big),\ 1\leq i\leq N$.

In light of the existing estimators that converge to an upper bound of the proportion of true nulls {(as discussed in Section~\ref{ssec:prop_est}), Theorem~\ref{prop_het}} makes an attempt to obtain an upper bound on the (asymptotic) global FDR for any such estimators in a heterogeneous setting. We introduce $\Delta = \sum_{i=1}^{N}q^{(i)} |d_i|$
with $d_i = r_0^{(i)} - r_0^*$, as a natural measure of network heterogeneity in terms of proportion of true nulls at different nodes, noting that $\sum_{i=1}^{N}q^{(i)}=1$ and $\sum_{i=1}^{N}q^{(i)} d_i=0$ according to the definition of $r_0^*$.  

We need one mild technical assumptions on $F_i$ to invoke the asymptotic characterization from~\cite{genovese2002operating} {via Lemma~\ref{het}}. 

\begin{assumption}
$F_i(t)$ is {continuously} differentiable {in a neighborhood of} $t=\widetilde{\tau}^{(i)}$, $F_i'\big(\widetilde{\tau}^{(i)}\big)\neq  \widetilde{\beta}^{(i)}$ all $i$, and $\widetilde{\tau}^{(i)}>0$ for at least one node.
\label{concav2}
\end{assumption}

\begin{theorem}\label{prop_het}
If Assumptions~\ref{concav2} holds, then for $r_0^*\leq\underset{1\leq i\leq N}{\min}\frac{1-\overline{r}_0^{\,(i)}}{1-r_0^{\,(i)}}$, we have
\begin{equation}\label{het_propbnd}\underset{m\rightarrow\infty}{\lim}\mathsf{FDR}\leq \alpha + \frac{\breve{V}+\breve{R}}{\big(\breve{R}-\Delta \big)^2}\Delta\, , \qquad \text{for } \Delta < \breve{R}\, ,
\end{equation}
where $\breve{V}=r_0^*\sum_{i=1}^{N}q^{(i)} \widetilde{\tau}^{(i)}$ and
\begin{align*}
    \breve{R}&=r_0^*\sum_{i=1}^{N}q^{(i)} \widetilde{\tau}^{(i)}+r_1^*\sum_{i=1}^{N} q^{(i)} F_i(\widetilde{\tau}^{(i)}).
\end{align*} 
\end{theorem}

\begin{proof}
Define $\underline{\beta}^*:=\beta(\alpha/r_0^*;r_0^*)$ and $\overline{\alpha}^{(i)} :=\big((1-r^{(i)}_0) \underline{\beta}^*+r^{(i)}_0\big)^{-1}$. Notice that ${\beta}\big(\overline{\alpha}^{(i)};r_0^{(i)}\big) =\underline{\beta}^*$. 

Recall that $\widetilde{\alpha}^{(i)}=\big({\big(1-\overline{r}_0^{\,(i)}\big)\overline{\beta}^*+\overline{r}_0^{\,(i)}}\big)^{-1}$. We now show $\widetilde{\alpha}^{(i)} < \overline{\alpha}^{(i)}$ for all $0<\alpha\leq 1$ and $1\leq i\leq N$. Observe, $\overline{\beta}^* > \frac{({1}/{\alpha})-r^*_0}{1-r^*_0}$ since $r^*_0< \overline{r}^*_0$. Therefore, $\widetilde{\alpha}^{(i)} < \overline\alpha^{(i)}$ (using the definitions of $\overline{\beta}^*$ and $\underline{\beta}^*$) is {implied by} to the following,
\begin{align*}
    &\Big(1-\overline{r}_0^{\,(i)}\Big)\bigg(\frac{(1/\alpha)-r^*_0}{1-r^*_0}\bigg)+\overline{r}_0^{\,(i)} > \\&\qquad\qquad\qquad\qquad\quad \Big(1-{r}_0^{\,(i)}\Big)\bigg(\frac{r^*_0\,(1/\alpha-1)}{1-r^*_0}\bigg)+{r}_0^{\,(i)}\nonumber\\
    &\iff \alpha > \frac{r^*_0\Big(1-{r}_0^{(i)}\Big)-\Big(1-\overline{r}_0^{\,(i)}\Big)}{\overline{r}_0^{(i)}-{r}_0^{\,(i)}},\ \text{for all } 0<\alpha\leq 1 \nonumber\\
    &\iff r^*_0\Big(1-{r}_0^{(i)}\Big)-\Big(1-\overline{r}_0^{\,(i)}\Big) \leq 0\\
    &\iff r_0^*\leq\frac{1-\overline{r}_0^{\,(i)}}{1-r_0^{\,(i)}}\ .
\end{align*}
Hence, we have
\begin{align*}
    &\frac{r_0^*\widetilde{\tau}^{(i)}}{r_0^*\widetilde{\tau}^{(i)}+r_1^*F_i(\widetilde{\tau}^{(i)})}=\frac{r_0^*}{r_0^*+r_1^* \widetilde{\beta}^{(i)}}\\
    &\le\frac{r_0^*}{r_0^*+r_1^* \underline{\beta}^*}
    = \alpha ,\quad \text{for all}\  \left\{1\leq i\leq N:\widetilde{\tau}^{(i)}>0\right\} \,,
\end{align*}
where the inequality follows from
\begin{equation*}
    \underline{\beta}^*={\beta}\big(\overline{\alpha}^{(i)};r_0^{(i)}\big) \leq {\beta}\big(\widetilde{\alpha}^{(i)};r_0^{(i)}\big) =\widetilde{\beta}^{(i)},
\end{equation*}
since $\widetilde{\alpha}^{(i)} < \overline{\alpha}^{(i)}$.
This implies,
\begin{equation}
    \frac{\sum_{i=1}^{N}q^{(i)} r_0^*\widetilde{\tau}^{(i)}}{\sum_{i=1}^{N}q^{(i)}\big(r_0^*\widetilde{\tau}^{(i)}+r_1^* F_i(\widetilde{\tau}^{(i)})\big)} \leq \alpha .\label{eq:bound_1}
\end{equation}
 {By Assumption~\ref{concav2} and Theorem~\ref{thm:plugin},} we can bound the asymptotic FDR as follows,
\begin{align}
    & \underset{m\rightarrow\infty}{\lim}\mathsf{FDR} =\frac{\sum_{i=1}^{N}q^{(i)}r_0^{(i)}\widetilde{\tau}^{(i)}}{\sum_{i=1}^{N}q^{(i)}\big(r_0^{(i)}\widetilde{\tau}^{(i)}+r_1^{(i)}F_i(\widetilde{\tau}^{(i)})\big)}\nonumber\\
    & = \frac{\sum_{i=1}^{N}q^{(i)}\big(r_0^*+ d_i\big)\widetilde{\tau}^{(i)}}{\sum_{i=1}^{N}q^{(i)}\Big(\big(r_0^*+d_i\big)\widetilde{\tau}^{(i)}+ \big(r^*_1-d_i\big)F_i(\widetilde{\tau}^{(i)})\Big)}\nonumber\\
    & = \frac{\sum_{i=1}^{N}q^{(i)} r_0^*\widetilde{\tau}^{(i)}+\sum_{i=1}^{N}q^{(i)} d_i\widetilde{\tau}^{(i)}}{\sum_{i=1}^{N}q^{(i)}A_i+\sum_{i=1}^{N}q^{(i)} B_i}\label{eq:limFDR},
    \end{align}
    where we define
    \begin{align*}
        A_i:= \big(r_0^*\widetilde{\tau}^{(i)}+ r^*_1 F_i(\widetilde{\tau}^{(i)})\big)\;\text{and }\;
       B_i:=d_i\big(\widetilde{\tau}^{(i)}-F_i(\widetilde{\tau}^{(i)})\big).
    \end{align*}
    Since $\widetilde{\tau}^{(i)}\le 1$ and $F(x)\le 1$, we can upper bound~\eqref{eq:limFDR} as 
    \begin{align}
    &\frac{\sum_{i=1}^{N}q^{(i)} r_0^*\widetilde{\tau}^{(i)}+\sum_{i=1}^{N}q^{(i)} |d_i|}{\sum_{i=1}^{N}q^{(i)}\big(r_0^*\widetilde{\tau}^{(i)}+ r^*_1 F_i(\widetilde{\tau}^{(i)})\big)-\sum_{i=1}^{N}q^{(i)} |d_i|}\nonumber\\
    & \overset{(a)}{\leq} \frac{\sum_{i=1}^{N}q^{(i)} r_0^*\widetilde{\tau}^{(i)}}{\sum_{i=1}^{N}q^{(i)}\big(r_0^*\widetilde{\tau}^{(i)}+r_1^* F_i(\widetilde{\tau}^{(i)})\big)} +\nonumber\\
    & \qquad \frac{\sum_{i=1}^{N}q^{(i)}\big(2 r_0^*\widetilde{\tau}^{(i)}+r_1^* F_i(\widetilde{\tau}^{(i)})\big)}{\Big(\sum_{i=1}^{N}q^{(i)}\big(r_0^*\widetilde{\tau}^{(i)}+r_1^* F_i(\widetilde{\tau}^{(i)})\big)-\Delta \Big)^2}\Delta\nonumber\\
    &\overset{(b)}{\leq} \alpha + \frac{\breve{V}+\breve{R}}{\big(\breve{R}-\Delta \big)^2}\Delta\ ,\label{boundincon}
\end{align}
where (a) follows from the definition of $\Delta$ and elementary algebra, and (b) from~\eqref{eq:bound_1} and the definitions of $\breve{V}$ and $\breve{R}$. 
\end{proof}

We observe that for consistent estimators of $r_0^{(i)}$, we have $\widetilde{\alpha}^{(i)}={\alpha}^{(i)}$, and as a result
\begin{align*}
    \widetilde{\beta}^{(i)}&={\beta}\big(\widetilde{\alpha}^{(i)};r_0^{(i)}\big) ={\beta}\big(\alpha^{(i)};r_0^{(i)}\big),\\ 
   \beta^*&=\beta(\alpha;r_0^*)\overset{(*)}{=}\beta(\alpha^{(i)};r_0^{(i)}),
\end{align*}
where ($*$) follows from the definition of $\alpha^{(i)}$. Hence, $\widetilde{\beta}^{(i)}=\beta^*$. Recall that $\widetilde{\tau}^{(i)}=\sup\{t:F_i(t)=\widetilde{\beta}^{(i)}t\}$. Using $\widetilde{\beta}^{(i)}=\beta^*$, we get  $\widetilde{\tau}^{(i)}=\sup\{t:F_i(t)={\beta}^*t\}=:\tau^{(i)}$ and $F_i(\tau^{(i)})=\beta^* {\tau}^{(i)}$. Now the bound \eqref{het_propbnd} will hold with $r_0^*\,\alpha\leq\alpha$ instead of $\alpha$, where $\breve{R}=(r_0^*+r_1^*\beta^*)\tau$, $\breve{V}=r_0^*\tau$, and $\tau = \sum_{i=1}^{N}q^{(i)} {\tau}^{(i)}$. 
\begin{corollary}\label{Fheter1-2} 
Under Assumptions~\ref{ass:consis} and~\ref{concav2}, we have
\begin{equation}\label{het_propbnd2}
    \underset{m\rightarrow\infty}{\lim}\mathsf{FDR}\leq r_0^*\alpha + \frac{\breve{V}+\breve{R}}{\big(\breve{R}-\Delta \big)^2}\Delta\, , \qquad \text{for } \Delta < \breve{R}\, .
\end{equation}
\end{corollary}

We now characterize the FDR loss (for consistent estimators) in a different way. Let $\widetilde{F}(t):=({1}/{r_1^*})\sum_{i=1}^N q^{(i)} r_1^{(i)} F_i(t)$ denote {the centralized alternative CDF and define $\tau^*:=\sup\{t:\widetilde{F}(t)=\beta^* t\}$. In order to capture the degree of heterogeneity in terms of the alternative distributions,} it is natural to control the distance {between $F_i$ and $\widetilde{F}$ at each node}. We formalize this by introducing a parameter $\delta_{F_i}$ below
along with a smoothness assumption on $\widetilde{F}(t)$. 
\begin{assumption}\label{Fdist}
$\underset{t_1\leq t\leq t_2}{\sup}|F_i(t)-\widetilde{F}(t)|\leq \delta_{F_i}$ for some $\delta_{F_i} >0$ where $t_1=\underset{i}{\min}\{{\tau}^{(i)}\}$ and $t_2=\underset{i}{\max}\{{\tau}^{(i)}\}$.
\end{assumption}
\begin{assumption}
$\widetilde{F}(t)$ is C-Lipschitz on the interval $[\underset{i}{\min}\,{\tau}^{(i)},\underset{i}{\max}\,{\tau}^{(i)}]$, i.e., $|\widetilde{F}(t)-\widetilde{F}(t')| \leq C|t-t'|$, for all $\underset{i}{\min}\,{\tau}^{(i)}\leq t,t'\leq \underset{i}{\max}\,{\tau}^{(i)}$.\label{lipschitz}
\end{assumption}
For instance, the one-sided p-value computed for $\Norm(\mu,1)$ statistic to test $\mathsf{H}_0:\mu = 0$ has Lipschitz constant $f_{\mu}(a) = e^{-\mu^2/2}e^{\,\mu Q^{-1}(a)}$ on interval $[a,b]$ where $Q(t)=1-\Phi(t)$ and $\Phi$ denotes the CDF of $\Norm(0,1)$ \cite{ray2011false}.
\begin{lemma}\label{tight}
Under Assumptions~\ref{Fdist} and~\ref{lipschitz}, if $\beta^*>C$, then $|{\tau}^{(i)}-\tau^*|\leq \frac{\delta_{F_i}}{\beta^*-C}$.
\end{lemma}
\begin{proof}
Recall that ${\tau}^{(i)}=\sup\{t:F_i(t)={\beta}^*t\}$ and $\tau^*=\sup\{t:\widetilde{F}(t)=\beta^* t\}$. Therefore, $\widetilde{F}(\tau^*)=\beta^*\tau^*$, $F_i({\tau}^{(i)})= {\beta}^*{\tau}^{(i)}$, and as a result
$\big|F_i({\tau}^{(i)})-\widetilde{F}(\tau^*)\big|=\beta^*|{\tau}^{(i)}-\tau^*|$. Now according to Assumptions \ref{Fdist} and \ref{lipschitz}, we get
\begin{align*}
    |{\tau}^{(i)}-\tau^*| &= 1/\beta^*\big|F_i({\tau}^{(i)})-\widetilde{F}(\tau^*)\big|\\
    &\leq \frac{1}{\beta^*}\Big(\big|\widetilde{F}(\tau^*)-\widetilde{F}({\tau}^{(i)})\big|+\big|\widetilde{F}({\tau}^{(i)})-F_i({\tau}^{(i)})\big|\Big)\\
    &\leq \frac{1}{\beta^*}\big(C\big|{\tau}^{(i)}-\tau^*\big|+\delta_{F_i}\big).
\end{align*}
Therefore, $\beta^*>C$ yields $|{\tau}^{(i)}-\tau^*|\leq \frac{\delta_{F_i}}{\beta^*-C}$.
\end{proof}

The condition $\beta^* > C$ is mild since we only need it to hold for the smallest Lipschitz constant on the interval $[\underset{i}{\min}\{{\tau}^{(i)}\},\underset{i}{\max}\{{\tau}^{(i)}\}]$.

\begin{theorem} \label{Fheter2}
Under Assumptions~\ref{ass:consis},~\ref{concav2},~\ref{Fdist} and~\ref{lipschitz}, if $\beta^* > C$, then 
\begin{equation*}
    \underset{m\rightarrow\infty}{\lim}\mathsf{FDR}\leq r_0^*\,\alpha + \frac{\check{V}+\check{R}}{\big(\check{R}-\Delta' \big)^2}\Delta'\, ,\quad \text{for } \Delta'< \check{R},
\end{equation*}
where $\check{R}=(r_0^*+r_1^*\beta^*)\tau^*$, $\check{V}=r_0^*\tau^*$, and $\Delta'=
\frac{1}{\beta^*-C}\sum_{i=1}^Nq^{(i)}(r_0^{(i)}+\beta^*r_1^{(i)})\delta_{F_i}$.
\end{theorem}
\begin{proof}
{According to Assumptions~\ref{ass:consis} and~\ref{concav2}},
\begin{align*}
    &\underset{m\rightarrow\infty}{\lim}\mathsf{FDR}=\frac{\sum_{i=1}^{N}q^{(i)}r_0^{(i)}{\tau}^{(i)}}{\sum_{i=1}^{N}q^{(i)}\big(r_0^{(i)}{\tau}^{(i)}+r_1^{(i)}F_i({\tau}^{(i)})\big)}\\
    &=\frac{\check{V}+\sum_{i=1}^{N}q^{(i)} r_0^{(i)}({\tau}^{(i)}-\tau^*)}{\check{R}+\sum_{i=1}^{N}q^{(i)}\big(r_0^{(i)}({\tau}^{(i)}-\tau^*)+r_1^{(i)}(F_i({\tau}^{(i)})-\widetilde{F}(\tau^*))\big)}\\
    &\overset{(a)}{\leq}\frac{\check{V}+\sum_{i=1}^{N}q^{(i)} r_0^{(i)}|{\tau}^{(i)}-\tau^*|}{\check{R}-\sum_{i=1}^{N}q^{(i)}\big(r_0^{(i)}|{\tau}^{(i)}-\tau^*|+r_1^{(i)}\beta^*|{\tau}^{(i)}-\tau^*|\big)}\ ,
\end{align*}
where (a) follows from $\big|F_i({\tau}^{(i)})-\widetilde{F}(\tau^*)\big|=\beta^*|{\tau}^{(i)}-\tau^*|$.

According to Lemma \ref{tight}, we have $|{\tau}^{(i)}-\tau^*|\leq \frac{\delta_{F_i}}{\beta^*-C}$ if $\beta^* > C$. 
Hence,
\begin{align*}
    \underset{m\rightarrow\infty}{\lim}\mathsf{FDR}&\leq 
    \frac{\check{V}+\Delta'}{\check{R}-\Delta'}
    \leq r_0^*\,\alpha + \frac{\check{V}+\check{R}}{\big(\check{R}-\Delta' \big)^2}\Delta'\ .
\end{align*}
\end{proof}
\begin{corollary}
Combining the results from Corollary \ref{Fheter1-2} and Theorem \ref{Fheter2} we get the following bound under the same assumptions as in Theorem \ref{Fheter2}.
\begin{equation*}
    \underset{m\rightarrow\infty}{\lim}\mathsf{FDR} \leq r_0^*\,\alpha + \min\Bigg\{\frac{\breve{V}+\breve{R}}{\big(\breve{R}-\Delta \big)^2}\Delta,\frac{\check{V}+\check{R}}{\big(\check{R}-\Delta' \big)^2}\Delta'\Bigg\} \ .
\end{equation*}
\end{corollary}
This result characterizes the impact of the two types of heterogeneity ($\Delta$ and $\Delta'$) on the asymptotic FDR. 
\begin{remark}
    Recall that $\beta^*=\beta(\alpha;r_0^*)	= \frac{(1/\a)-{r}_0^*}{1-{r}_0^*}$, which implies $r_0^*+r_1^*\beta^*=1/\alpha$. Therefore, we have $\Delta'\leq \frac{\underset{1\leq i \leq N}{\max}\delta_{F_i}}{\alpha(\beta^*-C)}$.
\end{remark}
We now study the asymptotic power under heterogeneity. Let $P^*=\widetilde{F}(\tau^*)$ denote the global BH asymptotic power. 
\begin{theorem}
Under the setting of Theorem \ref{Fheter2}, we have
\begin{equation*}
    \lim_{m\rightarrow\infty} \mathsf{power}\geq P^* - \min\left\{\frac{\underset{1\leq i \leq N}{\max}\delta_{F_i}}{1-C/\beta^*},\;\frac{\Delta'}{r_1^{*}}\right\}\, .
\end{equation*}
\end{theorem}
\begin{proof}
We start by noting that $\lim_{m\rightarrow\infty}\mathsf{power}$ equals to $\frac{1}{r_1^*}\sum_{i=1}^N q^{(i)} r_1^{(i)}F_i({\tau}^{(i)})$,
which can be rewritten and then lower bounded as follows,
\begin{align*}
    &\frac{1}{r_1^*}\sum_{i=1}^N q^{(i)} r_1^{(i)}\widetilde{F}(\tau^*)+\frac{1}{r_1^*}\sum_{i=1}^N q^{(i)} r_1^{(i)}\big(F_i({\tau}^{(i)})-\widetilde{F}(\tau^*)\big)\\
    &\geq P^* -\sum_{i=1}^N \frac{q^{(i)} r_1^{(i)}}{r_1^{*}}\beta^*\big|{\tau}^{(i)}-\tau^*\big|\overset{(a)}{\geq}  P^* - \frac{\sum_{i=1}^N \frac{q^{(i)} r_1^{(i)}}{r_1^{*}}\delta_{F_i}}{1-C/\beta^*}\\
    &\geq P^* - \min\left\{\frac{\underset{1\leq i \leq N}{\max}\delta_{F_i}}{1-C/\beta^*},\;\frac{\Delta'}{r_1^{*}}\right\}\, ,
\end{align*}
where (a) holds according to Lemma \ref{tight}.
\end{proof}

\vspace{-.5em}
\section{{Towards Optimality: Adaptive Segmentation}}
\label{sec:optimal}
By similar arguments as in \cite{genovese2002operating}, for target FDR $\alpha$ in a centralized (or pooled) setting, one can show that performing the BH procedure with size $\alpha/r_0^*$ is asymptotically optimal among the threshold procedures (i.e., procedures rejecting $[0,t]$ for some $t\leq 1$).
In this case the asymptotic rejection region is $[0,\mathcal{T}(\alpha/r_0^*)]$ where $\mathcal{T}(\alpha):=\sup\big\{t:\mathcal{G}(t)=(1/\alpha)\,t\big\}$. Following this direction, the authors in~\cite{ARIASCASTRO202142} proposed a \emph{scan procedure} that rejects the interval $[\sigma,\tau]$ asymptotically, where
\begin{align*}
    (\sigma,\tau) = \underset{(s,t)}{\arg\max}(t-s),\qquad
     \text{s.t.}\quad \frac{t-s}{\mathcal{G}(t)-\mathcal{G}(s)}\leq \alpha.
\end{align*}
This is a simpler method than the multiple reference points method~\cite{chi2007performance}, which selects multiple intervals in a delicate manner. 
Again we note that for the target FDR $\alpha$, performing the scan procedure with size $\alpha/r_0^*$ is asymptotically optimal among the procedures that reject a single interval, namely, procedures rejecting $[s,t]$ for some $0\leq s<t\leq1$. 
In this section we wish to present an explicit characterization of the asymptotic optimal rejection rule for general networks. We also provide a computation- and communication-efficient algorithm to approximate it in a star network.

\begin{assumption} \label{ass:pdf}
For all $1\leq i\leq N$, $F_i(0)=0$, $F_i(t)$ is continuously differentiable on $(0,1)$.
%, $f_i(t)=F_i'(t)>0$ almost everywhere on $t\in(0,1)$ 
$f_i(t)=F_i'(t)$ over $t\in(0,1)$ and $f_i(t)=0$ otherwise; $f_i(t)$  has a finite number of local extrema on $(0,1)$.
\end{assumption}
Let $\Pi$ denote the set of subsets of $[0,1]$ that can be written as a finite union of disjoint open intervals. Now consider 
\begin{equation}
    %[\Gamma^*(\alpha)]:=
    \underset{\Gamma\in\Pi^N}{\arg\max}\ \underset{m\rightarrow\infty}{\lim}\mathsf{power}(\Gamma), \ 
      \text{s.t.}\  \underset{m\rightarrow\infty}{\lim}\mathsf{FDR}(\Gamma)\leq \alpha. \label{opt}
\end{equation}
At first glance, it seems that solving \eqref{opt} (even over the simple threshold procedures) by scanning requires a search over $N$-dimensional space $[0,1]^N$ which is computationally exhaustive. However, as we shall see in Theorem~\ref{thm:optimal}, if the probability density functions at each node are known, a $1$-dimensional search would be sufficient to find the asymptotic optimal rejection region in the network.
 
Let $g_i=G_i'$ and $\nu$ denote the Lebesgue measure. Define,
\begin{align}
    \Gamma^{(i)}(t)&:=\left\{x\in(0,1):\frac{g_i(x)}{r_0^{(i)}}>t+1\right\},\label{eq:Gamma}\\
 c_\alpha&:=\inf\left\{t:\frac{\sum_{i=1}^N q^{(i)}r_0^{(i)}\nu(\Gamma^{(i)}(t))}{\sum_{i=1}^N q^{(i)}\int_{\Gamma^{(i)}(t)}g_i d\nu}\leq\alpha\right\}.\label{eq:cutoff}
\end{align}
\begin{theorem}\label{thm:optimal}
If Assumption~\ref{ass:pdf} holds, then $\Gamma^*(\alpha)=\left({\Gamma^*}^{(1)},\hdots,{\Gamma^*}^{(N)}\right)$ with ${\Gamma^*}^{(i)}:=\Gamma^{(i)}(c_\alpha)$, for all $1 \leq i \leq N$, is a solution to \eqref{opt} and it is unique up to a null set (i.e., a set of Lebesgue measure 0).
\end{theorem}
\begin{proof}
See Appendix~\ref{apx:opt}.
\end{proof}
It can be observed that solving for $c_\alpha$, requires a $1$-dimensional search over $t$.
Recall that $g_i(x) = r_0 + r_1 f_i(x)$. For $x\in (0,1)$, we have
\begin{align*}
    \Gamma^{(i)}(t)&= \left\{x:\frac{r_1^{(i)}}{r_0^{(i)}} f_i(x)> t \right\}
    = \left\{x: \Lambda_i(x) >\frac{r_0^{(i)}}{r_1^{(i)}} t \right\},
\end{align*}
where $\Lambda_i$ denotes the likelihood ratio at node $i$. A similar problem has been considered in \cite{cai2009simultaneous} and the same result has been established under the assumption that $f_i(t)$ is monotonically decreasing in $t$ for all $1\leq i\leq N$.

We note that if the density functions exist, then a decent estimation of $g_i$ (or $f_i$) together with Theorem~\ref{thm:optimal} as well as~\eqref{eq:Gamma} and~\eqref{eq:cutoff} would allow for the development of a computation- and communication-efficient algorithm for star networks. However, density estimation is undesirable and we wish to avoid it as much as we can. Therefore, we propose a computation- and communication-efficient algorithm to approximate the optimal solution without even requiring the existence of a density function. We will show that our procedure is a ``good" approximation of $\Gamma^*(\alpha)$ if it exists. {In order to achieve this, we define a fixed number of candidate intervals at each node and formulate the optimality problem over a finite set of feasible solutions.}

Fix some $\eps>0$. {Let $L^{(i)} := \eps/({q}^{(i)}{r}_0^{(i)})$ and ${K}^{(i)}(\eps):= \lfloor 1/L^{(i)}\rfloor$ {denote the length and the number of the candidate intervals at node $i$, respectively. The length of the candidate intervals at each node is calibrated to allow for a greedy solution as we shall see later in Theorem~\ref{thm:oracle-interval}.} Define ${\mathcal{A}}(\eps):=\big\{(i,j): 1\leq i \leq N,\ 1\leq j\leq {K}^{(i)}(\eps)\big\}$ and let ${I}_j^{(i)}:=L^{(i)}\cdot(j-1,j],\ (i,j)\in {\mathcal{A}}(\eps)$ denote the $j$-th interval at node $i$. Let $\Gamma_S:=\bigl(\Gamma_S^{(1)},\hdots,\Gamma_S^{(N)}\bigr)$ be the {tuple of regions} corresponding to a subset of intervals $S\subseteq\Ac(\eps)$, where $\Gamma_S^{(i)}:=\underset{{(a,b)\in S:\, a=i}}{\bigcup} I_b^{(a)}$.
We write down the FDR and power of the rejection region $\Gamma_S$ explicitly as below and simplify the expressions later,
  \begin{align*}
  \mathsf{power}(\Gamma_S)
     &=\mathbb{E}\Bigg[\frac{\sum_{(i,j)\in S}\sum_{k=1}^{m^{(i)}}\ind\Big\{\mathsf{H}_{0,k}^{(i)}=0, P_k^{(i)}\in {I}_j^{(i)}\Big\}}{m_1\vee 1}\Bigg],\\
     \mathsf{FDR}(\Gamma_S)&= \mathbb{E}\Bigg[\frac{\sum_{(i,j)\in S}\sum_{k=1}^{m^{(i)}}{\ind\Big\{\mathsf{H}_{0,k}^{(i)}=1,P_k^{(i)}\in\ {I}_j^{(i)}\Big\}}}{1\vee\sum_{(i,j)\in S}{\sum_{k=1}^{m^{(i)}}{\ind\big\{P_k^{(i)}\in{I}_j^{(i)}\big\}}}}\Bigg],
\end{align*}
where $\mathsf{H}_{0,k}^{(i)}=1$ denotes that the $k$-th null hypothesis at node $i$ is true, and similarly for $\mathsf{H}_{0,k}^{(i)}=0$.
We first aim to find the asymptotic optimal rejection rule among the procedures that reject intervals of the form $I_j^{(i)}$. Therefore, the asymptotic optimality problem is given as follows.
\begin{align}
    \underset{S\subseteq{\mathcal{A}}(\eps)}{\arg\max}\ \underset{m\rightarrow\infty}{\lim}\mathsf{power}(\Gamma_S), \ 
      \text{s.t.}\  \underset{m\rightarrow\infty}{\lim}\mathsf{FDR}(\Gamma_S)\leq \alpha.\label{seg_opt}
\end{align} 
By invoking the law of large numbers and the bounded convergence theorem, and recall that $m^{(i)}/m\xrightarrow{a.s.}q^{(i)}$, the asymptotic optimality problem can be expressed as 
     \begin{align*}
      &\Xi^*(\alpha,\eps):=\underset{S\subseteq{\mathcal{A}(\eps)}}{\arg\max} \ \frac{1}{r_1^*} \sum_{(i,j)\in S}q^{(i)}\mathbb{P}\left(\mathsf{H}_{0,1}^{(i)}=0,P_1^{(i)}\in\ {I}_j^{(i)}\right)\\
     & \text{s.t.} \quad \frac{\sum_{(i,j)\in S}q^{(i)}\,\mathbb{P}{\Big(\mathsf{H}_{0,1}^{(i)}=1,P_1^{(i)}\in\ {I}_j^{(i)}\Big)}}{\sum_{(i,j)\in S}q^{(i)}\,\mathbb{P}{\big(P_1^{(i)}\in{I}_j^{(i)}\big)}}\leq \alpha,\ \  S\neq \varnothing,
\end{align*}
where $S\neq \varnothing$ is needed due to the $\vee$ operation in $1\vee\sum_{(i,j)\in S}{\sum_{k=1}^{m^{(i)}}{\ind\{P_k^{(i)}\in{I}_j^{(i)}\}}}$.}
\begin{remark}
    It is worth mentioning that
    \begin{equation*}
        0 \leq \underset{m\rightarrow\infty}{\lim}\mathsf{FDR}(\Gamma_S)-\mathsf{FDR}(\Gamma_S)\le (1-p_S)^m,
    \end{equation*}
    \begin{equation*}
        0 \leq \underset{m\rightarrow\infty}{\lim}\mathsf{power}(\Gamma_S)-\mathsf{power}(\Gamma_S)\le (1-r_1^*)^m,
    \end{equation*}
    by~\cite[Theorem~1]{Storey03thepositive} where $p_S:=\sum_{(i,j)\in S}q^{(i)}\,\mathbb{P}{\big(P_1^{(i)}\in{I}_j^{(i)}\big)}$ is the probability of rejection for the set of selected intervals $\Gamma_S$. This tells us that the finite-sample objective function converges to the asymptotic one exponentially fast. 
\end{remark}
\vspace{-1.5em}
\subsection{{Oracle Solution}}\label{orac}
{Before giving a solution to optimization problem \eqref{seg_opt}, let us simplify our notation by adopting the following shorthands for evaluating CDFs on intervals,
\begin{align}
    G_i\left(I_j^{(i)}\right):&=G_i\left(j L^{(i)} \right)-G_i\left((j-1)L^{(i)}\right)\nonumber\\
    &= r_0^{(i)}U\left(I_j^{(i)}\right)+ r_1^{(i)} F_i\left(I_j^{(i)}\right)\label{eq:G},
\end{align}
where $F_i(I_j^{(i)}):=F_i\left(j L^{(i)}\right)-F_i\left((j-1) L^{(i)}\right)$ and  $U(I_j^{(i)}):= j L^{(i)} - (j-1) L^{(i)}=L^{(i)}$.
Notice that $\mathbb{P}{\big(P_1^{(i)}\in{I}_j^{(i)}\big)}=G_i(I_j^{(i)})$ and
\begin{align*}
    &\mathbb{P}\left(\mathsf{H}_{0,1}^{(i)}=0,P_1^{(i)}\in\ {I}_j^{(i)}\right)\\    &\hspace{0.1em}=\mathbb{P}\left(\mathsf{H}_{0,1}^{(i)}=0\right)\cdot\mathbb{P}\left(P_1^{(i)}\in\ {I}_j^{(i)}|\mathsf{H}_{0,1}^{(i)}=0\right)
   =r_1^{(i)}F_i\left(I_j^{(i)}\right).
\end{align*}
Similarly, $\mathbb{P}(\mathsf{H}_{0,1}^{(i)}=1,P_1^{(i)}\in\ {I}_j^{(i)})=r_0^{(i)}U(I_j^{(i)})$ with the uniform distribution under the null. Thus, we can express the problem in terms of the CDFs as follows,
\begin{align*}
    &\Xi^*(\alpha,\eps)=\underset{S\subseteq{\mathcal{A}(\eps)}}{\arg\max}\ \frac{1}{r_1^*}\sum_{(i,j)\in S}{r_1^{(i)}q^{(i)}F_i\left(I_j^{(i)}\right)}\nonumber\\
    &\hspace{4em} \text{s.t.}\quad \frac{\sum_{(i,j)\in S}{r_0^{(i)}q^{(i)}U\left(I_j^{(i)}\right)}}{\sum_{(i,j)\in S}q^{(i)}G_i\left(I_j^{(i)}\right)} \leq \alpha,\ \  S\neq \varnothing.\nonumber
\end{align*}
Now define {$h^{(i)}_j:= (1/\eps)\, q^{(i)}G_i(I_j^{(i)})$} and
\begin{equation*}
     {\mathcal{H}}:=\left\{{h}^{(i)}_j:1\leq i \leq N,\ 1\leq j\leq {K}^{(i)}\right\},
\end{equation*}
where we allow ties inside $\mathcal{H}$.
\begin{theorem}
$\Xi(\alpha,\eps)\subseteq \Xi^*(\alpha,\eps)$ for all $\alpha\in (0,1)$ where
\begin{align}
    &\Xi(\alpha,\eps) :=\left\{S\subseteq \mathcal{A}(\e):|S|=M^*,\ \sum_{(i,j)\in S}h_j^{(i)}=\sum_{\ell=1}^{M^*}h_{(\ell)}\right\},\nonumber\\
    & M^* := \max\Big\{0\leq M\leq |\mathcal{H}|: M \leq \alpha \sum_{\ell=1}^M h_{(\ell)}\Big\},\label{ms}
\end{align}
with ${h}_{(\ell)}$ being the $\ell$-th largest element of $\mathcal{H}$.\label{thm:oracle-interval}
\end{theorem}
\begin{proof}
Observe that $U(I_j^{(i)})=L^{(i)}=\eps/({q}^{(i)}{r}_0^{(i)})$. Hence, the optimization problem for $\Xi^*(\alpha,\eps)$ simplifies to
\begin{align}
    &\Xi^*(\alpha,\eps)= \underset{S\subseteq{\mathcal{A}}}{\arg\max}\sum_{(i,j)\in S}{r_1^{(i)}q^{(i)}F_i\left(I_j^{(i)}\right)}\label{eq:problem1}\\
    &\hspace{4em} \text{s.t.}\quad \frac{|S|}{\sum_{(i,j)\in S}h^{(i)}_j} \leq \alpha,\ \  S\neq \varnothing. \label{eq:problem2}
\end{align} 
We note that for any fixed $1\leq |S|=M\leq |\mathcal{A}|$, we have
\begin{align*}
    &\underset{{S\subseteq{\mathcal{A}}},\ {|S|=M}}{\arg\max}\sum_{(i,j)\in S}{r_1^{(i)}q^{(i)}F_i\left(I_j^{(i)}\right)}\\
    &=\underset{{S\subseteq{\mathcal{A}}},\ {|S|=M}}{\arg\max}M\eps+\sum_{(i,j)\in S}{r_1^{(i)}q^{(i)}F_i\left(I_j^{(i)}\right)}\\
    &\overset{(a)}{=}\underset{{S\subseteq{\mathcal{A}}},\ {|S|=M}}{\arg\max}\sum_{(i,j)\in S}{\left(\eps+r_1^{(i)}q^{(i)}F_i\left(I_j^{(i)}\right)\right)}\\
    &= \underset{{S\subseteq{\mathcal{A}}},\ {|S|=M}}{\arg\max}\sum_{(i,j)\in S}\left(r_0^{(i)}q^{(i)}U\left(I_j^{(i)}\right) +r_1^{(i)}q^{(i)}F_i\left(I_j^{(i)}\right)\right),
\end{align*}
where (a) follows from the condition that $|S|=M$. Thus the optimization problem in~\eqref{eq:problem1}
 and~\eqref{eq:problem2} for fixed $|S|=M$ can be further simplified as follows, 
 \begin{align*}
      \underset{{S\subseteq{\mathcal{A}}},\ {|S|=M}}{\arg\max}\sum_{(i,j)\in S}{h^{(i)}_j},
    \hspace{3em} \text{s.t.}\quad \frac{M}{\sum_{(i,j)\in S}h^{(i)}_j} \leq \alpha,
\end{align*}
where we use~\eqref{eq:G} and $h^{(i)}_j= (1/\eps) \,q^{(i)}G_i(I_j^{(i)})$. 
Observe that  
\begin{equation*}
    S\in\left\{S'\subseteq \mathcal{A}(\e):|S'|=M,\ \sum_{(i,j)\in S'}h_j^{(i)}=\sum_{\ell=1}^{M}h_{(\ell)}\right\}
\end{equation*}
maximizes $\sum_{(i,j)\in S}{h^{(i)}_j}$ and minimizes ${M}/\left({\sum_{(i,j)\in S}h^{(i)}_j}\right)$ simultaneously over the solutions with $|S|=M$. 
Since $\underset{{S\subseteq{\mathcal{A}}},\ {|S|=M}}{\max}\sum_{(i,j)\in S}{r_1^{(i)}q^{(i)}F_i(I_j^{(i)})}$ is non-decreasing in $M$,
\begin{align*}
    &\underset{S\subseteq{\mathcal{A}}}{\arg\max}\sum_{(i,j)\in S}{r_1^{(i)}q^{(i)}F_i(I_j^{(i)})}
    \hspace{1em} \text{s.t.}\quad {|S|}\leq \alpha{\sum_{(i,j)\in S}h^{(i)}_j}\\
    &\overset{(b)}{\supseteq}\left\{S\subseteq \mathcal{A}(\e):|S|=M^*,\ \sum_{(i,j)\in S}h_j^{(i)}=\sum_{\ell=1}^{M^*}h_{(\ell)}\right\}=\Xi.
 \end{align*}
\end{proof}
We note that if we assume $F_i(I_j^{(i)})>0$ for all $(i,j)\in {\mathcal{A}}(\eps)$, then $(b)$ holds with equality because $\underset{{S\subseteq{\mathcal{A}}},\ {|S|=M}}{\max}\sum_{(i,j)\in S}{r_1^{(i)}q^{(i)}F_i(I_j^{(i)})}$ would be strictly increasing in $M$. Furthermore, if we assume there are no ties among $h^{(i)}_j$'s, then $\Xi(\alpha,\eps)$ would become a singleton where $S^*:=\{(i,j):h^{(i)}_j\geq h_{({M^*})}\}$ is its single element.}

\vspace{-1em}
\subsection{Greedy Aggregation Procedure}
 In this section, we propose greedy aggregation, a procedure that is based on segmenting the $[0,1]$ interval at each node and aims to estimate the oracle solution we have obtained in Theorem~\ref{thm:oracle-interval}.
 We note that with one round of communication, each node $i$ will access the total number of p-values $m$ and can estimate $\hat{q}^{(i)}\hat{r}_0^{(i)}$ using $\hat{q}^{(i)}=m^{(i)}/m$ and methods from Section~\ref{ssec:prop_est}. 
 The length and number of the candidate intervals at node $i$ is {estimated} by $\widehat L^{(i)}: = \eps/(\hat{q}^{(i)}\hat{r}_0^{(i)})$ and $\widehat{K}^{(i)}:= \lfloor1/\widehat L^{(i)}\rfloor$, respectively.   
 Obviously, a node with $\widehat{K}^{(i)}=0$ has no candidate intervals for rejection.
 
 Recall that node $i$ possesses $m^{(i)}$ p-values $\{P_1^{(i)},...,P_{m^{(i)}}^{(i)}\}$. At node $i$, we introduce the ``density" of p-values over interval {$\widehat{I}_j^{(i)}:=\widehat L^{(i)}\cdot(j-1,\; j]$} for $1\leq j \leq \widehat{K}^{(i)}$ as,
\begin{equation}
    \hat{h}_j^{(i)}:=\frac{1}{\eps m}\sum_{k=1}^{m^{(i)}}\ind\left\{ P^{(i)}_k\in\widehat{I}_j^{(i)}\right\},\label{eq:d_ij}
\end{equation}
which is an estimator of $h_j^{(i)}$ defined in Section~\ref{orac}.
Roughly, the idea is as follows. In each round of communication, the center node looks for the interval with the highest density of p-values among all nodes, and notify the corresponding node to reject the hypotheses in that interval; this process goes on until the estimated $\mathsf{FDR}$ exceeds $\a$ or all the hypotheses in the network have been rejected. To distinguish different actions needed at each node, the center node sends $1$ (or $0$) to indicate that rejections (or no rejections) need to be made, while any local node send $-1$ to the center node if all of its p-values have been rejected. We now formally describe our method. We use the notation $\langle k \rangle$ to denote the $k$-th iteration of the algorithm.

\smallskip
\noindent{\bf\underline{First round:}}
Each node $i$ first computes 
\begin{align*}
    \hat{h}^{(i)}_{\max}\langle 1\rangle :=\underset{1\leq j \leq \widehat{K}^{(i)}}{\max}{\hat{h}_j^{(i)}}\hspace{1.2em}\text{and}\hspace{1.2em}
    j^{(i)}_{\langle 1\rangle} := \underset{1\leq j \leq \widehat{K}^{(i)}}{\arg\max}\,{\hat{h}_j^{(i)}},
\end{align*}
and then sends $\hat{h}^{(i)}_{\max}\langle 1\rangle$ to the center node. The center node estimates $\widehat{\mathsf{FDR}}=1/\hat{h}_{\max}\langle 1\rangle$ where $\hat{h}_{\max}\langle 1\rangle:=\underset{1\leq i\leq N}{\max}\hat{h}^{(i)}_{\max}\langle 1\rangle$. 

Case I: If $\widehat{\mathsf{FDR}}\leq \alpha$, then the center node sends $1$ to node $i^*_{\langle 1\rangle}:= \underset{1\leq i\leq N}{\arg\max}\,{\hat{h}_{\max}^{(i)}\langle 1\rangle}$, which rejects the interval 
\begin{align*}
    \widehat L^{(i^*_{\langle 1\rangle})}\cdot\big(j^*_{\langle 1\rangle}-1,\; j^*_{\langle 1\rangle}\big],
\end{align*}
where  $j^*_{\langle 1\rangle} := j^{(i^*)}_{\langle 1\rangle}$. When ties exist, i.e., several nodes achieve $\hat{h}_{\max}\langle 1\rangle$, the center node rejects only one of them arbitrarily. The center node sends a $0$ to the other nodes so that they know no rejections will be made on their local hypotheses.

Case II: If $\widehat{\mathsf{FDR}}> \alpha$, the center node terminates the procedure with no rejections.

\medskip

\noindent{\bf\underline{$k$-th round:}}
Node $i^*_{\langle k-1\rangle}$ (i.e., the node that received $1$ in the $(k-1)$-th round) updates the center node by sending
\begin{align*}
    \hat{h}^{(i^*_{\langle k-1\rangle})}_{\max}\langle k\rangle:=\underset{ j  \text{ not rejected}}{\max}{\hat{h}_j^{(i^*_{\langle k-1\rangle})}},
\end{align*}
or sending $-1$ if all the intervals at node $i^*_{\langle k-1\rangle}$ have already been rejected. For all the other nodes ($i\neq i^*_{\langle k-1\rangle}$), we have $\hat{h}^{(i)}_{\max}\langle k\rangle=\hat{h}^{(i)}_{\max}\langle k-1\rangle$. 

Case I: If $\underset{1\leq i\leq N}{\max}\hat{h}^{(i)}_{\max}\langle k\rangle>0$, the center node sets $\hat{h}_{\max}\langle k\rangle=\underset{1\leq i\leq N}{\max}\hat{h}^{(i)}_{\max}\langle k\rangle$ and estimates 
\begin{align*}
    \widehat{\mathsf{FDR}}=\frac{k}{\sum_{\ell=1}^k \hat{h}_{\max}\langle \ell\rangle}.
\end{align*}
If $\widehat{\mathsf{FDR}}\leq \alpha$, the center node rejects all the hypotheses in the interval corresponding to $\hat{h}_{\max}\langle k\rangle$ in the same way as in the first round.

Case II: If $\underset{1\leq i\leq N}{\max}\hat{h}^{(i)}_{\max}\langle k\rangle\in\{0,-1\}$ or $\widehat{\mathsf{FDR}} > \alpha$, the center node terminates the procedure. 

The communication cost of our greedy aggregation method is $\mathcal{O}(\log m)$ bits per iteration, since the nodes only need to send the number of p-values that lie in certain intervals according to~\eqref{eq:d_ij}. The number of communication rounds $Q$ is bounded by the total number of candidate intervals in the network, i.e., $Q\leq \sum_{i=1}^N \widehat{K}^{(i)}\leq\left\lfloor \frac{1}{\eps} \sum_{i=1}^N\hat{q}^{(i)}\hat{r}_0^{(i)}\right\rfloor$. Therefore, for fixed $\epsilon$, the total communication cost is $\mathcal{O}(\log m)$ bits. In practice, following the literature on histograms~\cite{davies1947statistical,freedman1981histogram}, we suggest taking $\eps$ proportional to $m^{-1/2}$ or $m^{-1/3}$ which results in communicating $\mathcal{O}(m^{1/2}\log m)$ and $\mathcal{O}(m^{1/3}\log m)$ bits respectively. This is in contrast with communicating $\Oc(m)$ real-valued p-values for centralized inference.
It is noteworthy that {(for a fixed $\eps$)} the number of communication rounds is asymptotically upper bounded by $\a/\e$ (see Remark~\ref{rem:upper}).

\vspace{-1em}
\subsection{Asymptotic Performance}
Let $\hat{h}_{(k)}$ denote the $k$-th largest element of
\begin{equation*}
    {\widehat{\mathcal{H}}}:=\left\{\hat{h}_j^{(i)}:1\leq i \leq N,\ 1\leq j\leq \widehat{K}^{(i)}\right\},
\end{equation*}
where we allow ties inside ${\widehat{\mathcal{H}}}$.
Note that our greedy procedure essentially rejects the intervals corresponding to the $\widehat{M}$ largest elements of ${\widehat{\mathcal{H}}}$, where
\begin{align*}
    \widehat{M}:&=\min\left\{1\leq M \leq |{\widehat{\mathcal{H}}}|:\frac{M}{\sum_{\ell=1}^M \hat{h}_{(\ell)}} >\alpha \right\}-1\\
    &=\min\left\{1\leq M \leq |{\widehat{\mathcal{H}}}|:\frac{1}{M}{\sum_{\ell=1}^M \hat{h}_{(\ell)}} < 1/\alpha \right\}-1
\end{align*}
and $\widehat{M}=|{\widehat{\mathcal{H}}}|$ if $\frac{1}{M}{\sum_{\ell=1}^M \hat{h}_{(\ell)}} \geq 1/\alpha$ for all $1\leq M \leq |{\widehat{\mathcal{H}}}|$. But since $\frac{1}{M}{\sum_{\ell=1}^M \hat{h}_{(\ell)}}$ is non-increasing on $1\leq M \leq |{\widehat{\mathcal{H}}}|$, $\widehat{M}$ can be written as
\begin{align*}
    \widehat{M}=\max\left\{0\leq M \leq |{\widehat{\mathcal{H}}}|:M \leq \alpha \sum_{\ell=1}^M \hat{h}_{(\ell)}\right\}.
\end{align*}
Thus the procedure rejects some (arbitrary) $\widehat{S}\in  \widehat{\Xi}(\alpha,\eps)$, where
\begin{align*}
    \widehat{\Xi}(\alpha,\eps):=\left\{S\subseteq \widehat{\mathcal{A}}(\eps):|S|=\widehat{M},\ \sum_{(i,j)\in S}\hat{h}_j^{(i)}=\sum_{\ell=1}^{\widehat{M}} \hat{h}_{(\ell)}\right\},
\end{align*}
where $\widehat{\mathcal{A}}(\eps):=\big\{(i,j): 1\leq i \leq N,\ 1\leq j\leq \widehat{K}^{(i)}(\eps)\big\}$.

Consider the optimal asymptotic power that can be achieved by rejection regions of the form $\Gamma_S$, $S\subseteq\Ac(\eps)$ subject to asymptotic FDR control, defined as
\begin{align*}
 \Pc_{\Xi^*}(\alpha,\eps):= 
    \underset{S\subseteq{\mathcal{A}}(\eps)}{\max}\ \underset{m\rightarrow\infty}{\lim}\mathsf{power}(\Gamma_S), \\
      \text{s.t.}\  \underset{m\rightarrow\infty}{\lim}\mathsf{FDR}(\Gamma_S)\leq \alpha,
\end{align*} 
Define $\widehat\Gamma_{S}^{(i)}$ and $\widehat\Gamma_{S}$, $S\subseteq\widehat\Ac(\eps)$ via $\widehat{I}_j^{(i)}$ similar to $\Gamma_{S}^{(i)}$ and $\Gamma_{S}$.
\begin{propo}\label{localoptim}
Assume $G_i(t),\ 1\leq i\leq N$ are continuous.  Fix some $\eps>0$ and $\alpha\in (0,1)$.
Under Assumption~\ref{ass:consis}, if $\alpha \sum_{\ell=1}^{M^*} h_{(\ell)}\notin \mathbb{N}$ and $1/L^{(i)}\notin \mathbb N$ for all $1\leq i\leq N$, then 
$\underset{m\rightarrow\infty}{\lim}\mathsf{FDR}(\widehat\Gamma_{\widehat S})\leq \alpha$ and $\mathsf{power}(\widehat\Gamma_{\widehat S})\to\Pc_{\Xi^*}$.
\end{propo}
\begin{proof}
{According to Assumption \ref{ass:consis} and the strong law of large numbers we get $\widehat L^{(i)}\xrightarrow{a.s.}L^{(i)}$ for a fixed $\eps>0$ and all $1\leq i\leq N$. Therefore, $\widehat{K}^{(i)}={K}^{(i)}$ and $\widehat{\Ac}(\e)=\Ac(\e),\ a.s.$ for large $m$ since $1/L^{(i)}\notin \mathbb N$. Let $\widehat{G}_i$ denote the empirical CDF of p-values at node $i$. We note,
\begin{equation}
    \hat{h}^{(i)}_j = \frac{m^{(i)}}{{\eps} m}\left[\widehat{G}_i\left(j\widehat L^{(i)}\right)-\widehat{G}_i\left((j-1)\widehat L^{(i)}\right)\right]\label{interval-conv}
\end{equation}
By the Glivenko-Cantelli theorem~\cite{smirnov1944approximate}, $\widehat L^{(i)}\xrightarrow{a.s.}L^{(i)}$, and continuity of $G_i$ we get}, $\left|\hat{h}^{(i)}_j-{h}^{(i)}_j\right|\xrightarrow{a.s.} 0$ for all $(i,j)\in\Ac(\e)$ and as a result, $\sum_{\ell=1}^{|\Ac(\e)|}\left|\hat{h}_{(\ell)}-h_{(\ell)}\right|\xrightarrow{a.s.} 0$. Observe,
\begin{align*}
    \underset{1\leq M\leq |\Ac(\e)|}{\max}\,\left|\sum_{\ell=1}^M \left(\hat{h}_{(\ell)}-h_{(\ell)}\right)\right|&\leq \underset{1\leq M\leq |\Ac(\e)|}{\max}\,\sum_{\ell=1}^M\left|\hat{h}_{(\ell)}-h_{(\ell)}\right|\\
    &=\sum_{\ell=1}^{|\Ac(\e)|}\left|\hat{h}_{(\ell)}-h_{(\ell)}\right|\xrightarrow{a.s.} 0.
\end{align*}
 Hence, $\widehat{M}=M^*,\ a.s.$ for large $m$ since $\alpha \sum_{\ell=1}^{M^*} h_{(\ell)}\notin \mathbb{N}$. Since $\underset{(i,j)\in\Ac(\e)}{\max}\,|\hat{h}^{(i)}_j-{h}^{(i)}_j|\xrightarrow{a.s.} 0$,  $\hat{h}^{(i)}_j$ and ${h}^{(i)}_j$ have the same order up to ties in $\Hc$, almost surely. {Therefore, $\widehat{\Xi}\subseteq\Xi\ a.s.$ which implies $\widehat{\Xi}\subseteq\Xi^*\ a.s.$ according to Theorem~\ref{thm:oracle-interval}. Hence, $\widehat S\in \Xi^*\ a.s.$ for large $m$. The result follows from $|\hat{h}^{(i)}_j-{h}^{(i)}_j|\xrightarrow{a.s.} 0$ and similar arguments as in~\eqref{interval-conv} for $F_i(t)$ and $U(t)$ together with bounded convergence theorem.}   
\end{proof}

\begin{corollary}
    Under the same setting as in Proposition~\ref{localoptim}. Suppose for each $i$, $G_i(t)$ is differentiable and $\left|\hat{r}_0^{(i)}-\overline{r}_0^{\,(i)}\right|=\Oc_p\left(u^{(i)}(m)\right)$. Then, $\mathsf{FDR}(\widehat\Gamma_{\widehat S})\leq \alpha+\Oc(z_m)$ and $\mathsf{power}(\widehat\Gamma_{\widehat S})=\Pc_{\Xi^*}+\Oc(z_m)$ where $z_m=u_m\vee m^{-1/2}$ with $u_m={\max}_i\ u^{(i)}(m)$.
\end{corollary}

\begin{remark}
\label{rem:upper}
From \eqref{ms}, we have $M^*\leq \alpha/\eps$, therefore the number of rejected intervals and consequently the number of algorithm rounds are asymptotically bounded above by $\alpha/\eps$. 
\end{remark}

Define $\Pc^*(\alpha):=\underset{m\rightarrow\infty}{\lim}\mathsf{power}(\Gamma^*(\alpha))$ as the (asymptotic) optimal power among all the methods that reject a finite union of disjoint intervals, 
where $\Gamma^*$ is defined in Theorem~\ref{thm:optimal}. The following theorem shows that the asymptotic performance of the greedy aggregation procedure is near-optimal for small $\eps$.
\begin{theorem}\label{asymopt}
Under Assumptions \ref{ass:pdf}, we have 
\begin{equation*}
    \underset{\delta\rightarrow 0}{\underline{\lim}}\ \underset{\eps\rightarrow 0}{\underline{\lim}}\ \Pc_{\Xi^*}(\alpha+\delta,\eps)\geq \Pc^*(\alpha).
\end{equation*}
\end{theorem}

\begin{proof}
Let $\overline\Gamma^*:=\bigcap_{\Gamma_S\supseteq\Breve{\Gamma}^*}\Gamma_S$ denote the minimal cover of ${\Breve{\Gamma}}^*:=\left({\Breve{\Gamma}}^*_1,\hdots,{\Breve{\Gamma}}^*_N\right)$ where 
${{\Breve{\Gamma}}^*}_i:={\Gamma^*}^{(i)}\cap (0,K^{(i)}L^{(i)}]$
using intervals of length $L^{(i)}=\e/(q^{(i)}r_0^{(i)})$ for each $i$. Now define $\overline{\Pc}^*:=\underset{m\rightarrow\infty}{\lim}\mathsf{power}(\overline\Gamma^*)$ and $\overline{\mathsf{FDR}}^*:=\underset{m\rightarrow\infty}{\lim}\mathsf{FDR}(\overline\Gamma^*)$. We note that for small enough $\eps$ we get,
\begin{align}
   &\overline{\mathsf{FDR}}^*= \frac{\sum_{i=1}^N q^{(i)}r_0^{(i)}U\left({\overline{\Gamma}^*}^{(i)}\right)}{\sum_{i=1}^N q^{(i)}G_i\left({\overline{\Gamma}^*}^{(i)}\right)}\nonumber\\
    &\leq \frac{2\e \kappa^*+\sum_{i=1}^N q^{(i)}r_0^{(i)}U\left({\Gamma^*}^{(i)}\right)}{\sum_{i=1}^N q^{(i)}G_i\left({\Gamma^*}^{(i)}\right)-\sum_{i=1}^N q^{(i)}G_i\left((K^{(i)}L^{(i)},1]\right)}\nonumber\\
    & \leq \alpha+ \frac{\sum_{i=1}^N q^{(i)}\left [r_0^{(i)}U\left({\Gamma^*}^{(i)}\right)+G_i\left({\Gamma^*}^{(i)}\right)\right ]}{\left(\sum_{i=1}^N q^{(i)}G_i\left({\Gamma^*}^{(i)}\right)-D(\eps)\right)^2}D(\eps)\label{FDRloss}
\end{align}
where $D(\eps)=\max\left(2\e \kappa^*,\sum_{i=1}^N q^{(i)}G_i\left((K^{(i)}L^{(i)},1]\right)\right)$, $\kappa^*=\sum_{i=1}^N J^*_i$ and $J^*_i$ denotes the total number of intervals in ${\Gamma^*}^{(i)}$. Notice that $D(\eps)\to 0$ as $\eps\to 0$ according to the absolute continuity of the Lebesgue integral.
Therefore, for all $\delta>0$, there exists $\eps_0(\delta)$ such that $\Pc_{\Xi^*}(\alpha+\delta,\eps)\geq \overline{\Pc}^*$ for all $\eps<\eps_0$,
where we have used the optimality of $ \Pc_{\Xi^*}$ in terms of power, and the fact that $\overline\Gamma^*$ will satisfy the FDR constraint for small enough $\eps$, according to \eqref{FDRloss}.
We observe that $\overline{\Pc}^*\geq\Pc^*(\a)-D'(\eps)$ where $D'(\eps)=(1/r_1^*)\sum_{i=1}^N q^{(i)}r_0^{(i)}F\left((K^{(i)}L^{(i)},1]\right)$. Hence, we get $\Pc_{\Xi^*}(\alpha+\delta,\eps)\geq {\Pc}^*(\alpha)-D'(\eps)$ for all $\eps<\eps_0(\delta)$ and $\delta>0$. Thus, by the absolute continuity of the Lebesgue integral, we get $\underset{\eps\rightarrow 0}{\underline{\lim}}\  \Pc_{\Xi^*}(\alpha+\delta,\eps)\geq {\Pc}^*(\alpha)$ for all $\delta > 0$ completing the proof.
\end{proof}

\vspace{-.5em}
\section{Simulations}
\label{sec:sim}
In this section, we evaluate the empirical performance of our algorithms. In all the experiments, we have $N=5$ nodes and the proportion of false null hypotheses at node $i$ is set as $r_1^{(i)}=(N-i+1)/(2N)=0.5-(i-1)/10$. We consider two distributions for the statistics: $\Norm(\mu,1)$ and $\text{Cauchy}(\mu,1)$ where $\mu$ denotes the location parameter. One-sided p-values are computed against $\mathsf{H}_0:\mu=0$. At node $i$ the statistics under the alternative $\mathsf{H}_1:\mu>0$ are generated according to $\mathcal{N}(\mu^{(i)},1)$ or $\text{Cauchy}(\mu^{(i)},1)$ where $\mu^{(i)}\sim \mathsf{Unif} \left[\mu_{\text{base}}^{(i)}-0.5, \mu_{\text{base}}^{(i)}+0.5\right]$. The overall network (asymptotic) target FDR is set to be $\alpha = 0.2$ and the FDR and power are estimated by averaging over $1000$ trials. 
 
We compare the performance of our methods with \emph{no-communication BH}, \emph{pooled BH}, and the \emph{optimal} rejection region. The no-communication BH simply refers to the method in which each node performs a local BH procedure with $\alpha^{(i)}=\alpha/\hat{r}_0^{(i)}$. {In order to estimate $r_0^{(i)}$, we adopt the spacing estimator (discussed in Section~\ref{ssec:prop_est}) with $s_m=m^{-0.7}$.} The pooled BH refers to the global BH procedure where all the p-values in the network are pooled and shuffled, and the BH procedure is carried out with target FDR $\alpha/\hat{r}_0^*$ over all the p-values from all nodes. The target FDR of the proportion matching method is set to $\alpha/\hat{r}_0^*$ as well. For the greedy aggregation method, we take $\eps = \alpha\, m^{-1/2}$ for Experiments 1 and 3, $\eps = \eta\,\alpha\, m^{-1/2}$ for Experiment 2 (a), and $\eps = 2.5\,\alpha\, m^{-1/2}$ for Experiment 2 (b-c). The optimal rejection regions are computed according to~\eqref{eq:Gamma} and~\eqref{eq:cutoff}.
 
We would like to point out that our experiment settings deviate from Assumption~\ref{ass:fixed_dist} considerably. The proportion matching method is basically built upon this assumption. Nevertheless, as we have discussed in Section~\ref{sec:robust}, we shall see in the experiments that the method is robust against heterogeneous alternatives and its performance essentially follows the pooled BH performance closely.

\smallskip
\noindent{{\bf Experiment 1 (vary m).} The statistics are generated according to the Gaussian distribution. We set $\mu_{\text{base}}^{(i)}=1.25\, i$, $m^{(i)}=({N-i+1})/{N}=(1-0.2(i-1))\cdot n$ and vary $n$ from $10^2$ to $10^5$. We observe that the greedy aggregation method has a near-optimal performance when $n$ is large. On the other hand, no-communication BH is preferred over the proportion matching method in this experiment.}

\smallskip
\noindent{{\bf Experiment 2 (vary $\mu$).} We fix {$m^{(i)}=(1-0.2(i-1))\cdot 10^3$}, and consider the following three cases:
%\smallskip

\noindent\textbf{(a)} Gaussian statistics, $\mu^{(i)}=\eta\, i$, and $\eta$ ranges from $0.5$ to $2$.
%\smallskip

\noindent\textbf{(b)} Cauchy statistics, $\mu^{(i)}=\mu$, and $\mu$ ranges from $2$ to $10$.
%\smallskip

\noindent\textbf{(c)} Cauchy statistics at nodes $1$, $3$, and $5$ and Gaussian statistics at nodes $2$ and $4$. We fix $\mu^{(i)}=\mu$, and $\mu$ ranges from $2$ to $5$.
%\smallskip

We observe that the performance of the greedy aggregation method is very close to the optimal solution in this setting. Also, the proportion matching method has a performance close to the pooled BH method.
}

\begin{figure}[h]
\centering
  \includegraphics[scale=0.45]{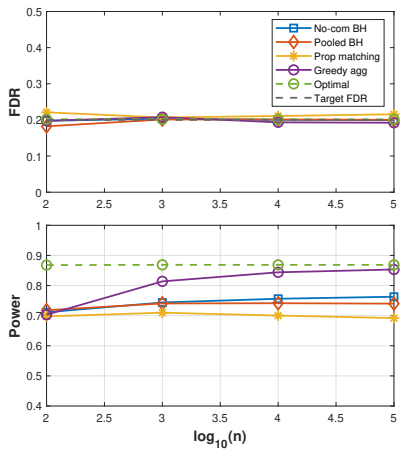}
    \vspace{-1em}
  \caption{Experiment~1: The effect of the number of p-values on the performance of a heterogeneous network with Gaussian statistics ($m^{(i)}=n(1.2-0.2i)$).}\vspace{-1em}
\end{figure}

\smallskip
\noindent{{\bf Experiment 3 (dependent p-values).} 
We set {$m^{(i)}=(1-0.2(i-1))\cdot 10^3$}, and generate Gaussian statistics with $\mu_{\text{base}}^{(i)}=1.25\, i$ and a tapering covariance structure where $\Sigma_{i,j} = \rho^{|i-j|}$ and $\rho$ ranges from $0$ to $0.9$. It can be observed that our methods show stable and consistent performance in this setting as well.}

\begin{figure*}[t]\centering
\begin{subfigure}{}\includegraphics[scale=0.45]{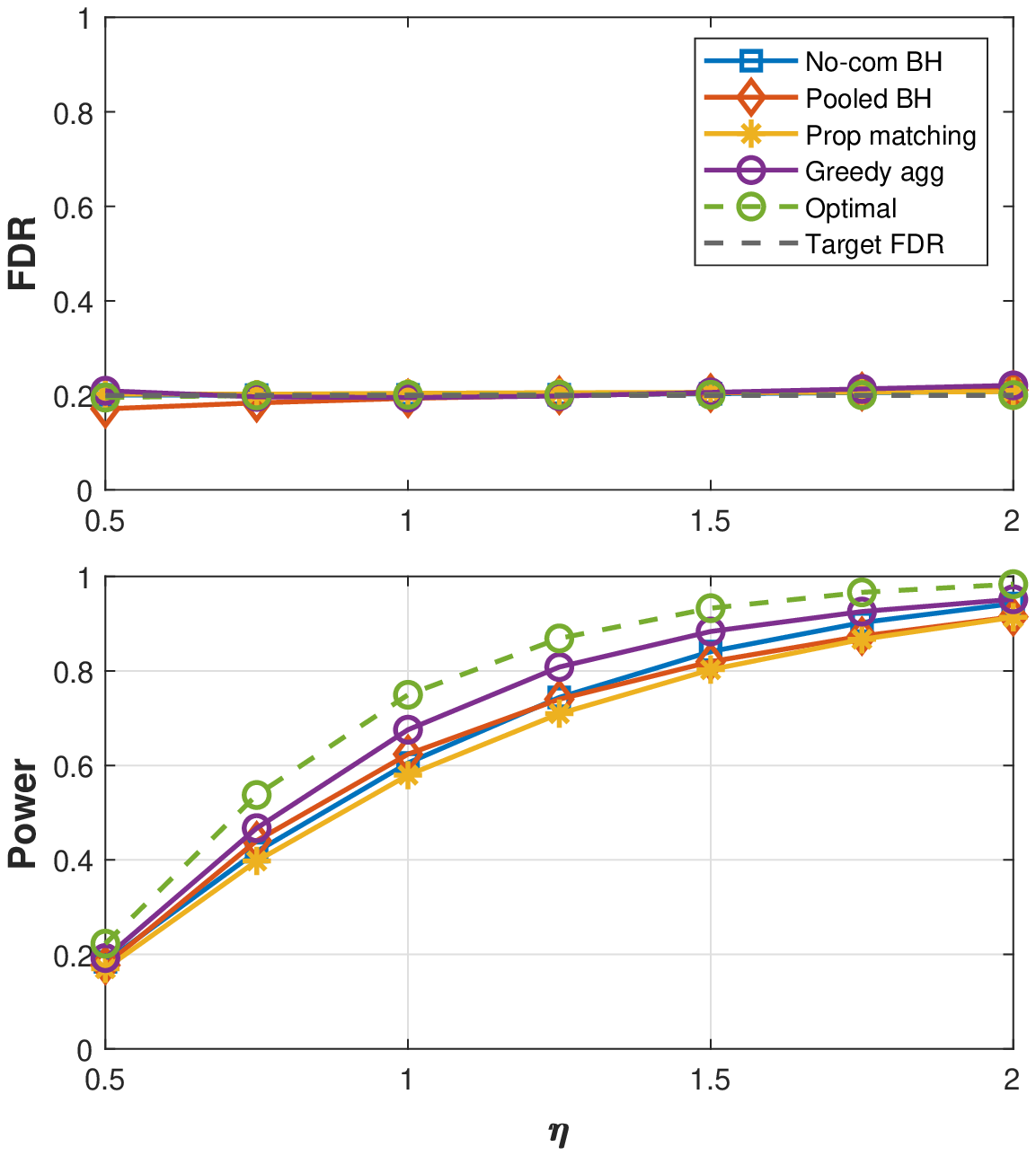}\end{subfigure}
\begin{subfigure}{}\includegraphics[scale=0.45]{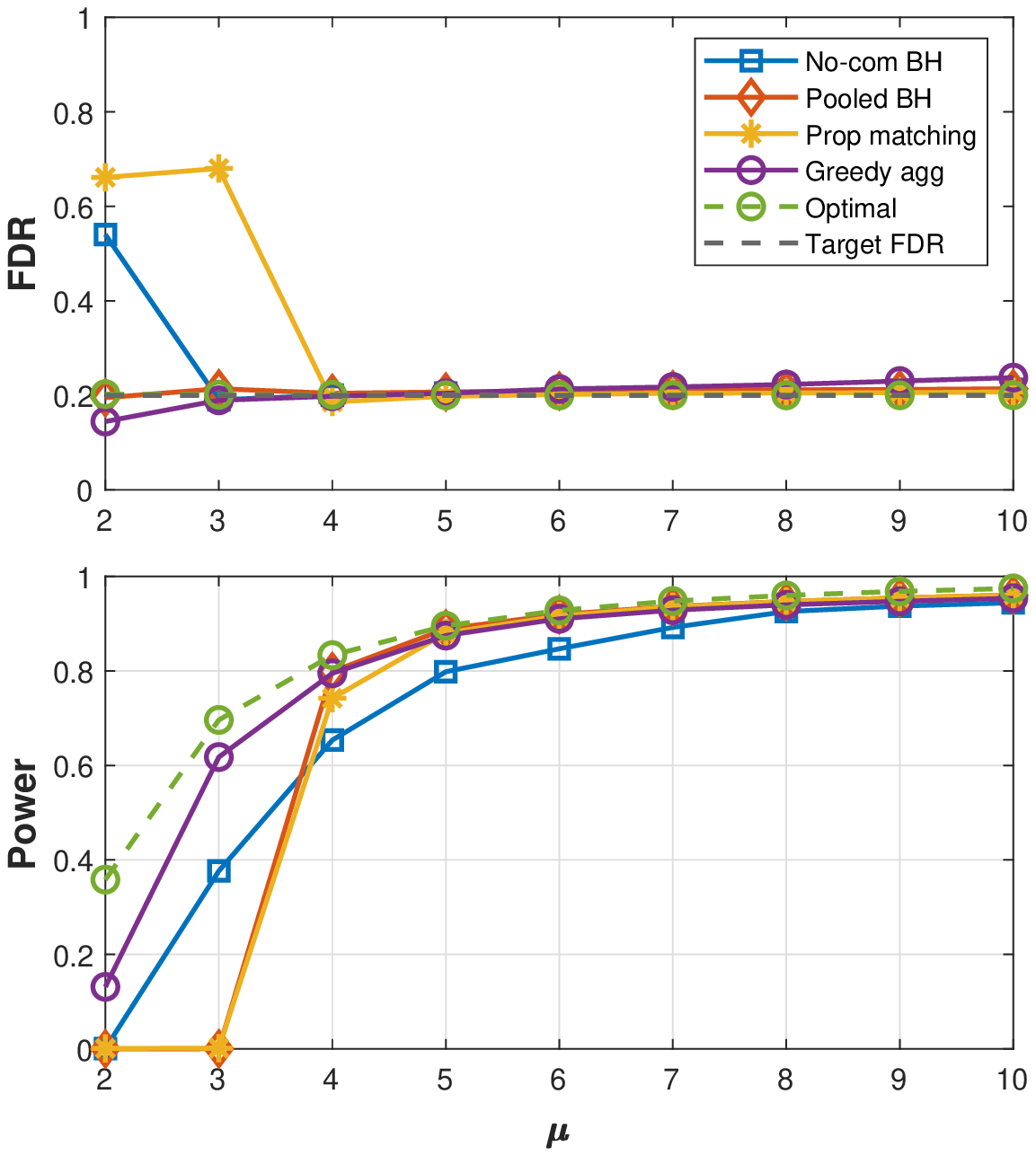}\end{subfigure}
\begin{subfigure}{}\includegraphics[scale=0.45]{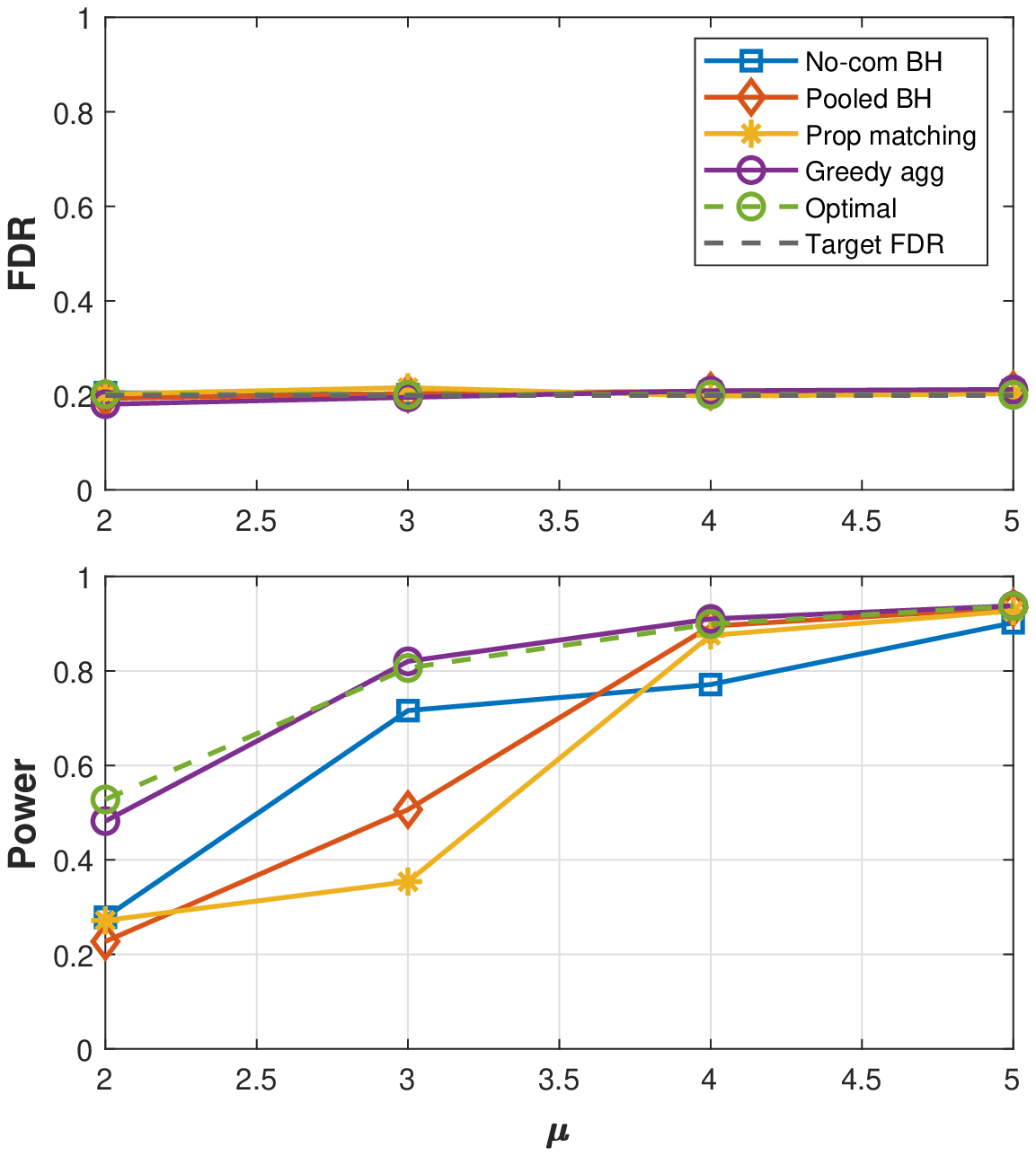}\end{subfigure}
\vspace{-2em}
\caption{From left to right, Experiment~$2$~(a): Gaussian statistics ($\underline{\mu}_{\text{base}}=\eta\cdot[1,2,3,4,5]$); Experiment~$2$~(b): Cauchy statistics ($\mu^{(i)}=\mu$); Experiment~$2$~(c): nodes $2$ and $4$ own p-values generated according to Gaussian statistics ($\mu^{(i)}=\mu$).}
\vspace{-1em}
\end{figure*}

\begin{figure}[h]
\centering
  \includegraphics[scale=0.45]{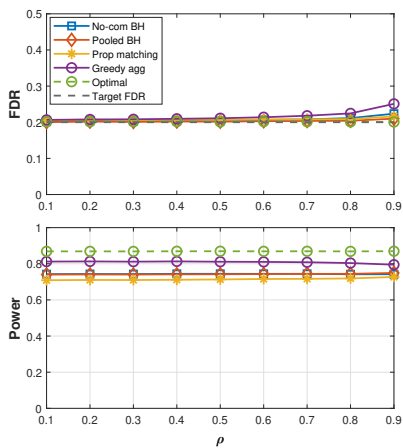}
    \vspace{-1em}
  \caption{Experiment~3: p-values generated according to Gaussian statistics with tapering correlation structure.}\vspace{-1em}
\end{figure}

\vspace{-.5em}

\section{Conclusion}
\label{sec:discuss}
We have taken an asymptotic approach to developing two methods for large-scale multiple testing problems in networks. Both methods come with low communication and computation costs, {asymptotic FDR control guarantee}, and competitive empirical performance. We have proposed the proportion matching method, a robust and communication-efficient algorithm that achieves the centralized BH performance asymptotically by adapting the local test sizes to the global proportion of true null hypotheses. We have also characterized the optimal (oracle) solution for the large-scale multiple testing problem in networks and proposed the greedy aggregation method as a communication- and computation-efficient algorithm to approximate the optimal rejection region effectively.

\vspace{-.5em}

\section{Acknowledgement}
\label{sec:ack}
The authors would like to thank the anonymous reviewers for their constructive comments {and the Associate Editor, Professor Lifeng Lai, for handling the submission}.

{\appendices

\vspace{-.5em}
\section{Technical Lemmas} \label{techlem} 
Consider applying the BH procedure to $m$ p-values generated $\iid$ according to $G(t;r_0)$ (defined in Section \ref{sec:prp-m}). The following lemma concerns relaxing the assumptions of the  Theorem~1 in~\cite{genovese2002operating}. In our proof, $F$ is not assumed to be concave and multiple solutions to $F(t) = \beta(\alpha;r_0)\, t$ can exist.
\begin{lemma}
Fix $\alpha\in(0,1)$. Define $\beta:=\beta(\alpha;r_0)$ and $\tau:=\tau(\alpha; r_0)=\sup\{t:F(t) = \beta\, t\}$. If $F(t)$ is {continuously} differentiable at $t=\tau$ and $F'(\tau)\neq {\beta}$, then the BH threshold satisfies {$|\tau_{\text{BH}}(\alpha)- \tau|=o_p((\log m)^{1/2+\gamma}/\sqrt{m})$ for any $\gamma > 0$ and $|\tau_{\text{BH}}(\alpha)- \tau|=o((\log m)^{1+\kappa}/\sqrt{m})\ a.s.$ for all $\kappa>0$.} \label{het}
\end{lemma}
\begin{proof}
We follow the approach in \cite{genovese2002operating} and only highlight the main differences. 
Let $\eps_m=m^{-1/4},\ m\geq 1$ and define, $a_m = \frac{m\tau}{\alpha}(1-\eps_m)$ and $b_m = \frac{m\tau}{\alpha}(1+\eps_m)$.
Recall that $D_{\text{BH}} = m/\alpha\, \tau_{\text{BH}}$ denote the BH deciding index, or equivalently, the number of rejections made by the BH procedure. First, we show $\mathbb{P}(D_{\text{BH}} > b_m)\rightarrow 0$ as follows,
\begin{align}
    &\mathbb{P}(D_{\text{BH}} > b_m)\nonumber \\
    &= \mathbb{P}\Bigg(\underset{k\,>\,b_m}{\bigcup}\big\{P_{(k)}\leq (k/m)\alpha\big\}\Bigg)\nonumber \\
    &= \mathbb{P}\Bigg(\underset{k\,>\,b_m}{\bigcup}\bigg\{\sum_{i=1}^m \ind\big\{P_i\leq (k/m)\alpha\big\}\geq k\bigg\}\Bigg)\nonumber\\
    &\leq \sum_{k\,>\,b_m}\mathbb{P}\Bigg[\bigg(\sum_{i=1}^m \ind\big\{P_i\leq (k/m)\alpha\big\}\bigg)-\mu(k)\geq k-\mu(k)\Bigg],\nonumber
\end{align}
where 
\begin{align}
   \mu(k) & =  \mathbb{E}\Big(\sum_{i=1}^m \ind\big\{P_i\leq (k/m)\alpha\big\}\Big)\nonumber\\
   & =r_0\,m\,(k/m)\alpha + (1-r_0)\,m F\big((k/m)\alpha\big).\nonumber
\end{align}
Recall that $\beta(\alpha;r_0) := \frac{(1/\a)-{r}_0}{1-{r}_0}$. We observe that
\begin{align}\label{diff}
     k-&\mu\,(k) = m r_1\Big[\frac{\alpha k}{m}\beta -F\big(\frac{\alpha k}{m}\big)\Big]=:h(k).
\end{align}
By Taylor's theorem, we obtain  
\begin{equation}
    h(b_m)= m r_1\Big[\tau\,\eps_m\big(\beta-F'(\tau)\big)+o(\eps_m)\Big]\nonumber.
\end{equation}
We observe,
\begin{align*}
    \underset{k > b_m}{\inf}h(k)/m &=  r_1\underset{k:\frac{\alpha k}{m}>\frac{\alpha b_m}{m}}{\inf} \Big[\frac{\alpha k}{m}\beta-F\big(\frac{\alpha k}{m}\big)\Big]\\
    &\geq r_1\underset{ t\geq (1+\eps_m)\tau}{\inf} \big(\beta t-F(t)\big)\ .
\end{align*}
We note that $\beta\,t-F(t)>0$ for all $t>\tau$ according to the definition of $\tau$. Notice that,
\begin{align*}
    r_1\Big[(1+\eps_m)\tau\beta-F((1+\eps_m)\tau)\Big]
    =\frac{h(b_m)}{m}\to 0.
\end{align*}
Also, $\beta t-F(t)$ is strictly increasing in a neighborhood of $\tau$ according to the continuous differentiability of $F(t)$ at $\tau$ and $\frac{d}{dt}\big(\beta t-F(t)\big)\Big|_{t=\tau}=\beta-F'(\tau) > 0$, where $F'(\tau) < \beta$ is deduced from $F'(\tau)\neq {\beta}$ and {the following result from~\cite{tarski1955lattice},} 
\[
    \sup\{t:F(t) = \beta\, t\}=\sup\{t:F(t) \geq \beta\, t\}.
\]
Hence, for large enough $m$ we get $\underset{k > b_m}{\inf}h(k)/m \geq {h(b_m)}/{m}$, and as a result $h(k)\geq h(b_m)$ for all $k>b_m$. According to $F'(\tau)<\beta$ and $\eps_m=m^{-1/4}$, we get $h(b_m)/\sqrt{m}\asymp m^{1/4}$.
Hence, by Hoeffding's inequality, we get
\begin{align}
    \mathbb{P}(D_{\text{BH}} > b_m) 
    &\leq \sum_{k\,>\,b_m}e^{-{2h(k)^2}/{m}}
    \precsim m\, e^{-{2h(b_m)^2}/{m}}\nonumber\\
    &\asymp m\, e^{-c \sqrt{m}}\rightarrow 0,\label{bm}
\end{align}
for some constant $c>0$. By the same argument as for $h(b_m)$ via Taylor's theorem and $F'(\tau)<\beta$, we get 
\begin{equation*}
    m^{-1/2}\big(\mu\,(a_m)-a_m\big)=-h(a_m)/\sqrt{m}\asymp m^{1/4} \ .
\end{equation*}
Hence, for some constant $c>0$, 
\begin{align}
    \mathbb{P}(D_{\text{BH}} < a_m) 
    &\leq \mathbb{P}\Big(P_{(a_m)}> (a_m/m)\alpha\Big)\nonumber \\
    &\leq \mathbb{P}\bigg(\sum_{i=1}^m \ind\big\{P_i\leq (a_m/m)\alpha\big\}< a_m\bigg)\nonumber\\
    &\overset{(a)}{\leq}e^{-{2h(a_m)^2}/{m}}\asymp e^{-c \sqrt{m}}\rightarrow 0,\label{am}
\end{align}
where $(a)$ follows from Hoeffding's inequality. Since the upper bounds in both~\eqref{bm} and~\eqref{am} are summable in $m$, we get $a_m\leq D_{\text{BH}}\leq b_m,\ a.s.$ for large $m$ by the Borel-Cantelli lemma. {The almost sure rate follows from the fact that the same argument holds for all $\{\eps_m\}$ decaying as (or slower than) $(\log m)^{1+\kappa}/\sqrt{m}$ for every $\kappa>0$. Regarding the convergence in probability rates, we note that the upper bounds in~\eqref{bm} and~\eqref{am} are not required to be summable in this case. Therefore, any $\{\eps_m\}$ decaying as (or slower than) $(\log m)^{1/2+\gamma}/\sqrt{m}$ works, concluding the claim.}
\end{proof}

\vspace{-1em}
\section{Proof of Theorem~\ref{thm:optimal}}\label{apx:opt}
Recall that $\Pi$ denotes the set of subsets of $[0,1]$ that can be written as a finite union of disjoint open interval, and $r_0^*=\sum_{i=1}^{N}{q^{(i)}\,r_0^{(i)}}$. Consider the map $\Tc:\Pi^N\to r_0^*\Pi$ given by
\begin{equation*}
    \left(\Gamma^{(1)},\hdots,\Gamma^{(N)}\right)\mapsto \bigcup_{i=1}^N\left\{\left\{\sum_{j=0}^{i-1} q^{(j)}r_0^{(j)}\right\}+q^{(i)}r_0^{(i)}\Gamma^{(i)}\right\},
\end{equation*}
where $q^{(0)}:=0$ and $r_0^{(0)}:=0$. In words, this map {normalizes and {merges} (potential) rejection regions from different nodes into an element of $r_0^*\Pi$.} It is straightforward to observe that this map is injective and the image of the map is
\begin{equation}
    \text{Im}(\Tc)=\left\{\Psi\mathbin{\big\backslash}\bigcup_{i=1}^{N-1}\left\{\sum_{j=1}^{i} q^{(j)}r_0^{(j)}\right\}:\Psi\in r_0^*\Pi\right\},\label{eq:Im}
\end{equation}
where we need to remove the endpoints of the intervals since $\Gamma^{(i)}$'s are unions of open intervals. 

Let $\Vc(\Gamma):=\underset{m\rightarrow\infty}{\lim}\frac{V(\Gamma)}{m}$ and $\Pc(\Gamma):=\underset{m\rightarrow\infty}{\lim}\mathsf{power}{(\Gamma)}$ denote the probability of false rejection and asymptotic power for a fixed $\Gamma\in \Pi^N$. Notice that we have
\begin{align}
    \Vc(\Gamma)&= \sum_{i=1}^N q^{(i)}r_0^{(i)}\nu(\Gamma^{(i)})
    =\sum_{i=1}^N \nu(q^{(i)}r_0^{(i)}\Gamma^{(i)})\nonumber\\
    &=\nu(\Tc(\Gamma)),\quad \text{for all } \Gamma\in \Pi^N,\label{eq:fdr}
\end{align}
where $\nu$ denotes the Lebesgue measure.
The asymptotic power can be expressed as follows,
\begin{align}
    \Pc(\Gamma)&=\frac{1}{r_1^*}\sum_{i=1}^N \int_{\Gamma^{(i)}}q^{(i)}r_1^{(i)}f_i(u)\, {\nu(du)}\nonumber\\
    &= \frac{1}{r_1^*}\sum_{i=1}^N \int_{q^{(i)}r_0^{(i)}\Gamma^{(i)}}\frac{q^{(i)}r_1^{(i)}}{q^{(i)}r_0^{(i)}} f_i\left(\frac{u'}{q^{(i)}r_0^{(i)}}\right)\, {\nu(du')}\nonumber\\
    &=\frac{1}{r_1^*}\sum_{i=1}^N \int_{q^{(i)}r_0^{(i)}\Gamma^{(i)}+\sum_{j=1}^{i-1} q^{(j)}r_0^{(j)}}\nonumber\\
    &\qquad\qquad\qquad\frac{r_1^{(i)}}{r_0^{(i)}} f_i\left(\frac{{t-\sum_{j=1}^{i-1} q^{(j)}r_0^{(j)}}}{q^{(i)}r_0^{(i)}}\right)\, {\nu(dt)}\nonumber\\
    &{=\frac{1}{r_1^*} \int_{\Tc(\Gamma)}\sum_{i=1}^N\frac{r_1^{(i)}}{r_0^{(i)}} f_i\left(\frac{{t-\sum_{j=1}^{i-1} q^{(j)}r_0^{(j)}}}{q^{(i)}r_0^{(i)}}\right)\, {\nu(dt)}}\nonumber\\
    &=\frac{1}{r_1^*} \int_{\Tc(\Gamma)}\Upsilon(t)\, {\nu(dt)},\label{eq:power}
\end{align}
with $\Upsilon(t) :=\sum_{i=1}^N \Upsilon_i(t)$ where
\begin{equation*}
    \Upsilon_i(t) := \frac{r_1^{(i)}}{r_0^{(i)}} f_i\left(\frac{t-\sum_{j=1}^{i-1} q^{(j)}r_0^{(j)}}{q^{(i)}r_0^{(i)}}\right).
\end{equation*}
}
Also note that, 
\begin{equation*}
    \underset{m\rightarrow\infty}{\lim}\mathsf{FDR}(\Gamma)=\begin{cases}\frac{\Vc(\Gamma)}{\Vc(\Gamma)+r_1^* \Pc(\Gamma)}\quad &\text{if }\Vc(\Gamma)>0,\\
    0 &\text{if }\Vc(\Gamma)=0.
    \end{cases}
\end{equation*}
For $\Psi\in r_0^*\Pi$, define
\begin{equation*}
    \Fc(\Psi):=\begin{cases}\frac{\nu(\Psi)}{\nu(\Psi)+\int_{\Psi}\Upsilon\, d\nu}\quad &\text{if }\nu(\Psi)>0,\\
    0 &\text{if }\nu(\Psi)=0.
    \end{cases}
\end{equation*}
Based on~\eqref{eq:fdr} and~\eqref{eq:power}, we now turn to solve the following optimization problem,
\begin{equation}
    \underset{\Psi\,\in\, r_0^*\Pi}{\arg\max}\ \int_{\Psi}\Upsilon\, d\nu, \quad
      \text{s.t.}\   \Fc(\Psi)\leq \alpha. \label{opt:trans}
\end{equation}
First, we introduce a set of intervals $\Psi(\tau)$ based on $\Upsilon(t)$, and show that $\Psi(\tau)$ can be used to represent any solution to~\eqref{opt:trans} up to a null set. Define,
\begin{align}
    \Psi(\tau):&=\left\{t\in(0,r_0^*):\Upsilon(t)>\tau\right\}.\label{eq:upsilon}
\end{align}
We note that under Assumption~\ref{ass:pdf}, $\Psi(\tau)\in r_0^*\Pi$, $\Psi(\tau)$ is decreasing in $\tau$, and $\nu(\Psi(\tau))$ is continuous in $\tau$ with $\nu(\Psi(0))=r_0^*$ and $\underset{\tau\rightarrow\infty}{\lim}\nu(\Psi(\tau))=0$. Therefore, for any fixed {$\Psi'\in r_0^*\Pi$} with $\nu(\Psi')=\nu_0$ there exists {at least} one $\tau'\geq 0$ such that $\nu(\Psi(\tau'))=\nu_0$. But clearly, {$\int_{\Psi(\tau')}\Upsilon\, d\nu >\int_{\Psi'}\Upsilon\, d\nu$ if $\nu\left(\Psi(\tau')\,\Delta\, \Psi'\right)>0$} because,
\begin{align*}
    \int_{\Psi(\tau')}\Upsilon\, d\nu &=\int_{\Psi(\tau')\cap \Psi'}\Upsilon\, d\nu + \int_{\Psi(\tau')\setminus\Psi'}\Upsilon\, d\nu\\
    &\overset{(a)}{>} \int_{\Psi(\tau')\cap \Psi'}\Upsilon\, d\nu + \int_{\Psi'\setminus\Psi(\tau')}\Upsilon\, d\nu
    =\int_{\Psi'}\Upsilon\, d\nu,
\end{align*}
where (a) follows from~\eqref{eq:upsilon} and $\nu(\Psi')=\nu(\Psi(\tau'))=\nu_0$. Therefore, for any solution $\Psi''$ to \eqref{opt:trans}, there exists a $\tau''$ such that $\nu\left(\Psi(\tau'')\,\Delta\, \Psi''\right)=0$. Hence, since $\Psi(\tau)$ is decreasing in $\tau$, $\Psi^*:=\Psi(\tau^*)$ is a maximizer of \eqref{opt:trans}, where
\begin{equation*}
    \tau^* = \inf\left\{\tau\geq 0:\Fc\left(\Psi(\tau)\right)\leq \alpha\right\}.
\end{equation*}
Observe that {according to Assumption~\ref{ass:pdf}}, $\int_{\Psi}\Upsilon\, d\nu < \int_{\Psi^*}\Upsilon\, d\nu$ for $\{\Psi\subset\Psi^*: \nu(\Psi\, \Delta\, \Psi^*)>0\}$ which implies that $\Psi^*$ is the unique maximizer of \eqref{opt:trans} up to null sets. {We note that according to Assumption~\ref{ass:pdf} and~\eqref{eq:Im}, $\Psi(\tau)\in \text{Im}(\Tc)$ for all $\tau\geq 0$.
Hence, $\Psi^*\in \text{Im}(\Tc)$.} Now let $\breve\Gamma=\Tc^{-1}(\breve\Psi),\ \breve{\Psi}\in\text{Im}(\Tc)$. Let $\Ic^{(i)}:=(\sum_{j=1}^{i-1} q^{(j)}r_0^{(j)}, \sum_{j=1}^{i} q^{(j)}r_0^{(j)})$ and notice that the rejection region at node $i$ is determined by
\begin{align*}
    &\breve\Gamma^{(i)}=\frac{\left\{\breve{\Psi}\bigcap \Ic^{(i)} \right\}- \sum_{j=1}^{i-1} q^{(j)}r_0^{(j)}}{q^{(i)}r_0^{(i)}}.
\end{align*}
Therefore, for $\breve{\Psi}_1,\breve{\Psi}_2\in\text{Im}(\Tc)$ we have
\begin{equation*}    \nu(\breve\Psi_1\Delta\breve\Psi_2)=0\implies\sum_{i=1}^N\nu(\breve\Gamma^{(i)}_1\Delta\breve\Gamma^{(i)}_2)=0,
\end{equation*}
where $\breve\Gamma_1=\Tc^{-1}(\breve\Psi_1)$ and $\breve\Gamma_2=\Tc^{-1}(\breve\Psi_2)$.
Hence, the maximizer of \eqref{opt} is unique up to null sets.

In order to compute $\Gamma^*$, again we note that $\Psi(\tau)\in \text{Im}(\Tc)$, all $\tau\geq 0$. Therefore, for $\Gamma(\tau)=\Tc^{-1}(\Psi(\tau))$, ${\Gamma}^{(i)}(\tau)$ equals to
\begin{align*}
    &\ \ \ \frac{\left\{t: \Upsilon_i(t)>\tau\right\}- \sum_{j=1}^{i-1} q^{(j)}r_0^{(j)}}{q^{(i)}r_0^{(i)}}\\
    &\hspace{3em}=\left\{t:\frac{r_1^{(i)}}{r_0^{(i)}} f_i(t)>\tau\right\}\overset{(b)}{=}\left\{x:\frac{g_i(x)}{r_0^{(i)}}>\tau+1\right\},
\end{align*}
where (b) follows from the definitions of $f_i$ and $g_i$. Also, $\tau^*$ can be computed as follows,
\begin{align*}
    \tau^* &= \inf\left\{\tau\geq 0:\frac{\Vc(\Gamma(\tau))}{\Vc(\Gamma(\tau))+r_1^* \Pc(\Gamma(\tau))}\leq \alpha\right\}\\
    &=\inf\left\{\tau\geq 0:\frac{\sum_{i=1}^N q^{(i)}r_0^{(i)}\nu\left(\Gamma^{(i)}(\tau)\right)}{\sum_{i=1}^N q^{(i)}\int_{\Gamma^{(i)}(\tau)}g_i d\nu}\leq\alpha\right\}.
\end{align*}
Thus $\Gamma^*=\Tc^{-1}(\Psi(\tau^*))=\Gamma(\tau^*)$ and the proof is complete.

\section{Distribution functions used in Section \ref{sec:sim} (Simulations)}\label{ap: dists}
The CDF and PDF of a one-sided p-value computed for $\Norm(\mu,1)$ statistic with respect to $\mathsf{H}_0:\Norm(0,1)$ is given as follows~\cite{ray2011false} for $t\in (0,1)$,
\begin{equation*}
    F_{G,\mu}(t) = Q(Q^{-1}(t)-\mu)\ \text{ and }\  f_{G,\mu}(t) = e^{-\mu^2/2}e^{\,\mu Q^{-1}(t)},
\end{equation*}
% \begin{equation*}
%     f_{G,\mu}(t) = e^{-\mu^2/2}e^{\,\mu Q^{-1}(t)},\qquad 0\leq t\leq 1 ,
% \end{equation*}
where $Q(t)=1-\Phi(t)$ and $\Phi$ denotes the CDF of $\Norm(0,1)$.

Let $H_{C,\mu}(t)={1}/{\pi}\tan^{-1}(t-\mu)+1/2$ denote the CDF of a Cauchy statistic with location $\mu$ and scale $\sigma=1$. The CDF and PDF of a one-sided p-value computed for $\text{Cauchy}(\mu,1)$ statistic with respect to $\mathsf{H}_0:\text{Cauchy}(0,1)$ is given as follows.
\begin{align*}
    F_{C,\mu}(t) &= 1-H_{C,\mu}\left((1-H_{C,0})^{-1}(t)\right)\\
    &=1-H_{C,\mu}(\cot(\pi t))
    =\frac{1}{2}-\frac{1}{\pi}\tan^{-1}(\cot(\pi t)-\mu),\\
    f_{C,\mu}(t)&=\frac{d}{dt}F_{C,\mu}(t)=\frac{\cot^2(\pi t)+1}{(\cot(\pi t)-\mu)^2+1}.
\end{align*}

% \vspace{-0.6em}

\balance
\bibliographystyle{IEEEtran}
\bibliography{ref.bib}

% \begin{IEEEbiographynophoto}{Mehrdad Pournaderi}
% (S’21) received his B.S. degree in Electrical Engineering from the University of Tehran, Iran, in 2019 and the degree of Master of Statistics (Mathematics) from the University of Utah, USA, in 2022. He is currently a Ph.D. student in the Department of Electrical and Computer Engineering at the University of Utah.

% \begin{IEEEbiographynophoto}{Yu Xiang}
% (S’10–M’15) received his B.E. degree with the highest distinction in Telecommunications Engineering from Xidian University, China, and his Ph.D. degree in Electrical and Computer Engineering from the University of California, San Diego. He was a postdoctoral fellow at Harvard University, School of Engineering and Applied Sciences. He is currently an assistant professor in the Department of Electrical and Computer Engineering at the University of Utah. His research interests include statistical signal processing, information theory, and machine learning.

\end{document}